\documentclass[fontsize=10pt]{article}
\usepackage[hyphens]{url}
\usepackage{geometry}
\usepackage{fullpage}
\usepackage{graphicx}
\usepackage{natbib}
\usepackage[font=small,labelfont=bf]{caption}

\usepackage{mathtools}
\usepackage{amsthm}
\usepackage[ruled]{algorithm2e}

\usepackage{amsfonts}
\usepackage{amssymb}
\usepackage{booktabs}

\usepackage[hidelinks]{hyperref}

\usepackage{cleveref}

\usepackage{authblk}

\usepackage{import}
\usepackage{preamble}

\title{Designing Rules to Pick a Rule: Aggregation by Consistency}

\author[1,2]{Ratip Emin Berker}
\author[3]{Ben Armstrong\footnote{This work was primarily completed while affiliated with the University of Waterloo.}}
\author[1,2,4]{Vincent Conitzer}
\author[1]{Nihar B. Shah}

\affil[1]{Carnegie Mellon University}
\affil[2]{Foundations of Cooperative AI Lab (FOCAL)}
\affil[3]{Tulane University}
\affil[4]{University of Oxford}
\affil[ ]{\texttt{\{rberker, conitzer, nihars\}@cs.cmu.edu, research@benarmstrong.ca}}

\date{} 

\begin{document}

\maketitle

\begin{abstract}
    Given a set of items and a set of evaluators who all individually rank them, how do we aggregate these evaluations into a single societal ranking? Work in social choice and statistics has produced many aggregation methods for this problem, each with its desirable properties, but also with its limitations. Further, existing impossibility results rule out designing a single method that achieves every property of interest. Faced with this trade-off between incompatible desiderata, how do we decide which aggregation rule to use, \emph{i.e.}, what is a good \emph{rule picking rule}?
    
    In this paper, we formally address this question by introducing a novel framework for rule picking rules (RPRs). We then design a data-driven RPR that identifies the best aggregation method for each specific setting, without assuming any generative model. The principle behind our RPR is to pick the rule---among a set of ``acceptable'' rules that satisfy basic requirements---which maximizes the consistency of the output ranking if the data collection process were repeated. We introduce several consistency-related axioms for RPRs and show that our method satisfies them, including those failed by a wide class of natural RPRs. While we prove that the algorithmic problem of maximizing consistency is computationally hard, we provide a sampling-based implementation of our RPR that is efficient in practice. We run this implementation on known statistical models and find that, when possible, our method selects the maximum likelihood estimator of the data. Finally, we show that our RPR can be used in many real-world settings---such as peer review, contests, or political elections---to gain insights about how the rule currently being used can be modified or replaced to substantially improve the consistency of the process. 
    
    Taken together, our work bridges an important gap between the axiomatic and statistical approaches to rank aggregation, laying a robust theoretical and computational foundation for principled rule picking.
\end{abstract}

\section{Introduction}\label{sec:intro}

Suppose you have a collection of items, and evaluators who individually rank them. Finding the best method for aggregating such data into a single ranking of items is an age-old problem. But what makes an aggregation method \emph{good}? 

One common approach to this question, known as the \emph{axiomatic approach} in the social choice literature~\citep{Plott76:Axiomatic}, is to first select certain criteria (called \emph{axioms}) that the aggregation should satisfy, and then design rules that meet these axioms. However, celebrated impossibility results~\citep{Arrow63:Social,Gibbard73:Manipulation,Satterthwaite75:Strategy} prove even some fundamental axioms cannot be simultaneously satisfied, thereby eliminating all hope for a single ``ideal'' rule that fulfills every reasonable desideratum. On the other hand, even if the chosen axioms are satisfiable at once, there may be many rules that do so, making the selection among them somewhat arbitrary. 

Another common approach for picking the aggregation method, which we will refer to as the \emph{statistical approach}, is to treat the rankings as noisy estimates of an objective ground truth. By assuming a noise model (\emph{e.g.}, Plackett-Luce~\citep{Plackett75:Analysis,Luce59:Individual}, Mallows~[\citeyear{Mallows57:Non}]) for the data generating process, the aggregate ranking can be chosen as the one that maximizes the likelihood of the data under the assumed model. A key challenge is that the assumed model may not be accurate, which is commonly addressed with considering multiple generative models and then choosing one via cross validation~\citep{Zucchini00:Introduction}. However, irrespective of the model-selection process, the assumption of a generative model with a ground truth may be fundamentally flawed, especially in settings with legitimate differences of opinion, such as AI alignment~\citep{Ge24:Axioms}. Further, many natural voting rules with desirable properties are not the maximum likelihood estimator (MLE) for \emph{any} noise model~\citep{Conitzer05:Common,Conitzer09:Preference} and are thereby precluded by the statistical approach.

Given this vast array of tools from statistics and social choice, and no clear way of \emph{a priori} selecting from them, a natural question emerges: \emph{In a given setting, how do we pick which rule to use? In other words, what makes a good rule picking rule (RPR)}?

In this paper, we address this question. Unlike previous literature that largely focuses on picking an aggregate ranking, we want to explicitly pick a rule. Such an approach has several benefits: First, employing an RPR naturally leads to better interpretability of the process by providing a formal justification of why other rules (under which the winners could be different) were not adopted. This can mitigate malicious behavior by preventing a powerful actor (such as an incumbent running for reelection) from changing the rules of the game for their benefit. Second, different (perfectly reasonable) rules may be appropriate for different settings with different requirements, and RPRs offer a principled way of deciding which rule is the most appropriate for a given setting. Lastly, as we will see, our novel framework allows designing natural RPRs that can choose from \emph{any} set of rules, making it easy to continually incorporate novel rules into the aggregation process without changing the rule selection method.

~\\\noindent{\bf Our contributions are as follows:}
    \begin{enumerate}
        \item We introduce a novel framework for formally defining \emph{rule picking rules (RPR)} (\Cref{sec:prelim}). Our framework allows designing principled ways of adopting a rule appropriate for the data, without committing to a set of axioms or a generative model \emph{a priori}. 
        \item Inspired by prior work emphasizing the link between consistency and quality in related settings, we introduce our own RPR, \emph{Aggregation by Consistency ($\aba$)}, with the explicit goal of maximizing consistency if the data collection process were repeated (\Cref{sec:aba}). 
        \item We define several natural axioms for RPRs, and prove $\aba$ satisfies them, including those failed by a wide class of RPRs. For two axioms that $\aba$ fails, we prove impossibility results showing each of them are incompatible with the axioms $\aba$ satisfies (\Cref{sec:axiom}).
        \item We prove that the computational problem of checking if ``complete'' consistency be achieved for a given input (\emph{i.e.}, picking a rule that produces the exact same output) is $\NP$-complete (\Cref{sec:computational}).
        \item Nevertheless, we provide an implementation of $\aba$ that is efficient in practice and performs well in experiments on known distributions (\Cref{sec:experiments}). While our implementation makes no assumptions about the data generation process, we demonstrate experimentally that when a ground truth in fact exists, $\aba$ picks the rule which generates the ground truth ranking with maximum likelihood.
        \item We show that our approach can be applied to both score- and rank-based evaluations across a large variety of empirical settings, at times improving significantly upon the consistency given by rules used in practice. As such, $\aba$ can be used to guide decisions in these settings about which rule to adopt.
    \end{enumerate}
An implementation of our rule was awarded as one of the four winners in the \emph{2nd Computational Social Choice Competition} at the 33rd International Joint Conference on Artificial Intelligence (IJCAI 2024). In addition to checking axiomatic properties, the competition scored rules based on a specific welfare function that mapped the rules' outputs to utilities for voters, which was announced prior to the competition. Remarkably, even though we did not perform any tuning or optimization for this welfare function, our rule achieved 99.71\% of the total utility attained by a competing method explicitly designed to maximize this specific function. This result underscores the robustness and strong general performance of $\aba$ across diverse evaluation criteria.

Taken together, our work lays a robust theoretical and computational foundation for principled rule picking, paving the way for future work in this novel model. Our code and experiments are available \href{https://github.com/beng341/AgreementAggregation}{here}.


\section{Overview of the Proposed Approach} \label{sec:highlevel}

Before introducing our formal framework we discuss the basic idea of our approach. Suppose we have two sets of voters who independently rank the same items. A popular measure of the ``quality'' of an aggregation method is the \emph{consistency} between its output on two independent sets of rankings. As such, we seek to design an RPR that picks the rule that maximizes this consistency.\footnote{We will restrict our RPR to pick from rules that satisfy some basic axiomatic properties such as neutrality, which dictates that all items being ranked are treated equally, \emph{i.e.}, permuting the items in evaluators' rankings should result in the output ranking being accordingly permuted. This avoids the pathological case of a constant function with maximal consistency.} This intuition is inspired by a number of related settings where past work has emphasized the importance of consistent outputs:
\paragraph{(1) Peer review} A classical setting in which evaluator rankings need to be aggregated is peer review, and much work has been done to optimize this process; see~\citet{Shah22:Overview} for an overview. Many experiments split reviewers into two panels evaluating the same data (\citet{Jecmen22:Near} show the merits of this methodology), where interpanel disagreement is interpreted as a shortcoming~\citep{Obrecht07:Examining,Fogelholm12:Panel,Pier17:Your,Bast20:How}. For example, in the NeurIPS 2014 and 2021 conferences, the panels disagreed on more than half of the accepted papers, which is taken as a sign of arbitrariness in the review process~\citep{Lawrence14:NIPS,Cortes21:Inconsistency,Beygelzimer23:Has}. Similar experiments use interpanel consistency to compare distributed peer review with an expert panel~\citep{Patat19:Distributed,Kerzendorf20:Distributed}. Overall, it is clear that the peer review community already cares about consistency as an indicator for quality. 

\paragraph{(2) Clustering} One could view the task of rule picking as \emph{learning} a rule from ranking data. Unlike earlier work on learning voting rules that assumes black-box access to a ground-truth rule~\citep{Procaccia09:Learnability}, our setting is unsupervised and hence is more closely related to clustering problems. Indeed,~\citet{Ailon05:Aggregating} show that clustering and rank aggregation are closely linked and can be approached with near-identical algorithms. Similar to the problem of rule picking, an important challenge in clustering is that of model selection (\emph{e.g.}, the number of clusters). Previous work has shown that picking the model in a way that maximizes stability---\emph{i.e.}, would obtain similar results on several data sets from the same underlying source---yields results with higher accuracy; see~\citet{vonLuxburg10:Clustering} for an overview of theoretical results in clustering stability. Again, we see that the consistency and the quality of the output are closely related in this setting. 

\paragraph{(3) Minimum-variance unbiased estimator (MVUE)} A statistical estimator is \emph{unbiased} if its expected value equals the true value of the parameter being estimated. Among all unbiased estimators, the estimator with the smallest mean squared error is the one with the lowest variance. Indeed, MVUEs are very commonly studied in statistics~\citep{Rao49:Sufficent,Chapman51:Minimum} and related bounds for unbiased estimators are frequently used for rank aggregation~\citep{Hajek14:Minimax,Khetan16:Data}. Consider applying an unbiased estimator to two i.i.d.\ data samples. Since the variance of a random variable is exactly half of the expected squared difference between its two i.i.d.\ copies, the estimator that leads to the smallest expected error between the two data sets is once again the one with the minimum variance, \emph{i.e.}, the MVUE. This motivates the maximum consistency aggregation in our setting, where unbiasedness is interpreted as basic constraints for any acceptable rule, which we achieve by restricting our RPR to pick from neutral and anonymous rules (\emph{i.e.}, those that treat all items and evaluators the same, respectively). 

\paragraph{(4) AI alignment} There is nascent interest in applying tools from social choice to AI alignment processes, such as reinforcement learning from human feedback (RLHF)~\citep{Conitzer24:Social,Dai23:Mapping,Mishra23:AI}, with particular emphasis on consistency. Much like peer review, RLHF inevitably relies on input from a limited set of evaluators, despite aiming for broad societal alignment. As such, aggregation methods that are robust to repetitions of the process can decrease the arbitrariness of the final AI model due to the choice of evaluators. As pointed out by~\citet{Conitzer24:Social}, the focus of social choice on producing consistent aggregations helps make it an appropriate tool for this setting.

While these four settings (peer review, clustering, MVUEs, and RLHF) motivate an RPR that maximizes the consistency between two independent evaluations, it is not \emph{ex ante} clear how one would implement that. After all, in practice, we often have only one copy of the process. To explain how we circumvent this, \textbf{\Cref{alg:aba} introduces our RPR, named \emph{Aggregation by Consistency (AbC)}}. Note that this is an informal description, and a formal definition is provided subsequently in \Cref{sec:aba} (\Cref{def:aba}).

\begin{algorithm}
\caption{Aggregation by Consistency ($\aba$) }\label{alg:aba}
    \KwIn{A set of evaluations over items, a set of acceptable (``candidate'') rules} 
    \KwOut{A chosen rule, to be used for aggregating the evaluations} 
    1. Split the evaluators uniformly at random into two groups, considering each group as a copy of the process (in line with the peer review experiments above)\;
    2. For each candidate rule, compute the rule's outputs separately on the two groups, and measure the disagreement between these two outputs\;
    3. Output the rule (among candidate rules) that minimizes this disagreement.
\end{algorithm}

Importantly, \Cref{alg:aba} is agnostic to the types of input and output of the rules it is picking among; it only requires that we define a measure of disagreement for comparing the outputs of the rules (Step 2). Our approach is therefore broadly applicable to a wide variety of settings with specific evaluation formats---\emph{e.g.} rankings, ratings/scores, approval sets---and specific desired outputs---\emph{e.g.}, an aggregate ranking, a single winner, a committee of winners. In this paper, we focus on rules that output an aggregate ranking.

As we will show, $\aba$ has a number of advantages. First, it satisfies several important axioms for RPRs, even those that are failed by a wide class of natural RPRs.   Second, $\aba$ does not assume that the data is coming from a ground truth; however, if a generative model is indeed a reasonable approximation of the data, then, as we show in our experiments, the associated MLE yields a high consistency across random splits. $\aba$ then chooses this estimator, thus obtaining the benefits of the statistical approach. Third, several important social choice axioms can be imposed on $\aba$ simply by restricting the set of candidate rules to those that satisfy them, thus obtaining the benefits of the axiomatic approach.  Lastly, any implementation of our method is easy to continually extend by simply implementing new rules and adding them to our set of candidate rules.

In order to formally define $\aba$, we next introduce our novel framework for rule picking rules.


\section{Preliminaries and Rule Picking Rules}\label{sec:prelim}

We consider a set of \textit{voters} $\voters =\{1,2,\ldots, \nvoters\}$ each of whom individually evaluate a finite set of \textit{alternatives} $\cand$. For our axiomatic analysis (\Cref{sec:axiom}), it will be helpful to assume each voter provides a ranking of \emph{all} alternatives, although our implementation easily extends to other evaluations formats, as demonstrated by our experiments where voters provide partial rankings (\Cref{sec:gt_v_disagreement,sec:alma_rank}) or ratings/scores (\Cref{sec:score_data}).

A \emph{strict ranking} of alternatives is a total ranking of the elements in $\cand$, whereas a \emph{weak ranking} allows for ties. More formally, a weak ranking of $\cand$ is a complete and transitive binary relationship on $\cand$; a strict ranking is a weak ranking that is also asymmetric. We denote the set of all strict rankings of $\cand$ by $\tL(\cand)$, and the set of all weak rankings of $\cand$ by $\wL(\cand)$. For any (strict or weak) ranking $r$, we write $a \succ_r b$ if $a$ is strictly ranked above $b$ in $r$, and $b \succeq_r a$ otherwise (\emph{i.e.}, if $b$ is ranked weakly above $a$). We assume each voter $i \in \voters$ has a (strict) ranking $\sigma_i \in \tL(\cand)$ over the alternatives $A$. A \textit{(preference) profile} $\profile \in \tL(\cand)^\nvoters$ comprises the rankings of all voters.

\paragraph{Candidate rules}  A \emph{social welfare function (SWF)} is a mapping $f$ that, given a profile $\profile$, outputs a single weak ranking.\footnote{The reader might wonder why we focus on functions that get strict rankings as inputs but output a weak ranking. Restricting the output to strict rankings would require sometimes outputting multiple rankings in order to satisfy neutrality and anonymity (\emph{i.e.}, treating all alternatives and voters the same, respectively), \emph{e.g.}, in cyclic profiles. Using SWFs that return weak rankings (as in \citet[\oldS III]{Arrow63:Social}, \citet[\oldS 1.2.2]{Brandt16:Handbook}) allows us to  return a single ranking without violating these basic axioms.} To aggregate the rankings of the voters, we are interested in picking an SWF to use from a set of acceptable ``candidate'' rules. Importantly, our framework does not place any restrictions on the set of candidate rules. Indeed, our evaluations in \Cref{sec:experiments} will comprise various SWFs including MLEs of known noise models (\emph{e.g.}, the Plackett-Luce~\citep{Plackett75:Analysis,Luce59:Individual} and Mallows~[\citeyear{Mallows57:Non}] models) as well as various SWFs from social choice literature. One such class of SWFs we will sometimes pay attention to is the set of all (monotonic) positional scoring rules. 

\begin{defn}[Positional scoring rules]\label{def:pos} A (monotonic) \emph{positional scoring rule} is an SWF $f_s$ associated to a specified vector  ${s}=(s_j)_{j \in \{1,2,\ldots,\ncand\}}$ with $1=s_1\geq s_2 \geq \ldots \geq s_\ncand=0$. Given a profile $\profile$ as input, the output of $f_s$ is computed as follows:  for each alternative $a \in \cand$ and index $j \in  \{1,2,\ldots,\ncand\}$, let $\cmatrix_{\profile}[a,j]$ denote the number of voters in $\profile$ that rank alternative $a$ in the $j^\text{th}$ position. Further, define the total score of alternative $a$ in $\profile$ (with respect to vector $s$) as $t^{{s}}_{\profile}[a]=\sum_{j=1}^\ncand s_j\cmatrix_{\profile}[a,j]$. Then, on input profile $\profile$, the  SWF $f_{{s}}$  ranks the alternatives in decreasing total score, \emph{i.e.}, $a \succ_{f_{\boldsymbol{s}}(\profile)} b$ if and only if $t^{{s}}_{\profile}[a] > t^{{s}}_{\profile}[b]$. We denote the set of all (monotonic) positional scoring rules by $F_S$.
\end{defn}

In addition to being succinctly representable and easy to compute, positional scoring rules encapsulate many well-known SWFs, including \emph{plurality} ($f_p$, for $p=(1,0,\ldots,0)$), \emph{veto} ($f_v$, for $v=(1,\ldots,1,0)$), and \emph{Borda count} ($f_b$, for $b_j = \frac{\ncand-j}{\ncand-1}$).

\paragraph{Rule Picking Rules} We now introduce a novel framework for picking rules. Given a set of candidate rules (\emph{e.g.}, those that are the MLE of a noise model, or voting rules that satisfy certain axioms), we will define a function that maps a profile to the candidate rule(s) most appropriate for it.

\begin{defn}[Rule picking rules]
    A \emph{rule picking rule (RPR)} is a function $\rpr$ that given a set of SWFs $F$ (called \emph{candidate rules}) 
    and a profile $\profile$, outputs a subset of the candidate rules $Z(F,\profile) \subseteq \rules$. 
\end{defn}

An RPR offers a principled way to pick an SWF that should be used to aggregate the rankings in $\profile$. Ideally, an RPR will capture the idea that different (perfectly reasonable) SWFs may be appropriate for different profiles, \emph{e.g.}, because the profile resembles a specific distribution, or certain positions in the rankings give more information than others. In \Cref{sec:axiom}, we formalize these situations and define some natural axioms for RPRs.

To avoid ambiguity, ideally we would like to have $|\rpr(F,\profile)|=1$. However, just as many natural SWFs (\emph{e.g.}, positional scoring rules) can inevitably lead to ties among alternatives, as we will see, natural RPRs may lead to ties among SWFs for certain profiles, leading to $|\rpr(F, \profile)|>1$. In such cases, a tie-breaking order over $\rules$ can be used to pick the single rule to be adopted.


\section{Aggregation by Consistency (\emph{AbC})}\label{sec:aba}

Having introduced our framework of RPRs, we now formally present our proposed method. As discussed in \Cref{sec:highlevel}, we want our RPR to maximize the consistency between two independent sets of evaluations. Our approach is general can be implemented with any measure of distance comparing the outputs of a rule on the two groups of voters (to be used in Step 2 of \Cref{alg:aba}). For concreteness in our analysis and experiments, we turn to one such well-known measure between rankings.

\begin{defn}[Kendall-Tau distance with ties]\label{def:kt} Given a pair of weak rankings $r_1,r_2 \in \wL(\cand)$ and alternatives $a, b \in 
\cand$, let $ D_{r_1,r_2}^{a,b}$ be indicator variable that the rankings $r_1$ and $r_2$ strictly disagree about how alternatives $a$ and $b$ should be ordered. Similarly, let $T^{a,b}_{r_1,r_2}$ be the indicator variable that alternatives $a$ and $b$ are tied in at least one of $r_1$ or $r_2$. More formally,
\begin{align*}
    D_{r_1,r_2}^{a,b} &= \mathbb{I}[(a \succ_{r_1}b\text{ and }b \succ_{r_2} a)\text{ or }(b \succ_{r_1}a\text{ and }a \succ_{r_2} b)], \quad \text{and}\\
    T_{r_1,r_2}^{a,b} &= \mathbb{I}[(a \succeq_{r_1}b\text{ and }b \succeq_{r_1} a)\text{ or }(a \succeq_{r_2}b\text{ and }b \succeq_{r_2} a)]. 
\end{align*}
Then, the \emph{Kendall-Tau distance (with ties)} between $r_1$ and $r_2$ is
\begin{align}
    KT(r_1,r_2) = \sum_{ \{a,b\} \in \cand^2: a \neq b}\left( D_{r_1,r_2}^{a,b} + \frac{1}{2}T^{a,b}_{r_1,r_2}\right). \label{eq:kt}
\end{align}
\end{defn} 

It is worth noting that $KT$ is more traditionally defined for strict rankings, without the $T_{r_1,r_2}^{a,b}$ term. We explicitly add the term for ties in order to penalize indecisiveness; otherwise, a rule that always returns the empty ranking (all alternatives tied) would achieve zero disagreement (and hence maximum consistency), without violating neutrality or anonymity.  Weighing ties by $\frac{1}{2}$ is inspired by Kendall's Tau-b correlation coefficient~\citep{Kendall45:Treatment}. We now introduce our consistency-based RPR, which we formally define in \Cref{def:aba}.

\begin{rprfigure}{Aggregation by Consistency ($\aba$)}{def:aba}
    Given a profile $\profile$, consider the following random process: Initialize two sets $\voters_1 = \voters_2 = \emptyset$. For each voter $i \in \voters$, pick an index $j \in \onetwo$ uniformly at random and set $\voters_j \leftarrow \voters_j \cup \{i\}$. Let $\profile^{(j)}= (\sigma_i)_{i \in \voters_j}$ (\emph{i.e.}, the restriction of profile $\profile$ to voters in $\voters_j$) for each index $j \in \onetwo$. Then, given a set of candidate rules $F$, \emph{Agreement-based Aggregation ($\aba$)} is an RPR defined as
    \begin{align}\label{eq:aba}
        \aba(F, \profile)= \argmin_{f \in F}~~~\expc\left[KT\left(f(\profile^{(1)}),f(\profile^{(2)})\right)\right],
    \end{align}
    where the expectation is taken over splitting the profile $\profile$ into $\profile^{(1)}$ and $\profile^{(2)}$ by the random process described above.

\end{rprfigure}
 In words, $\aba$ returns the SWFs among $F$ that minimize, in expectation,\footnote{\label{fn:emptysplit} There is a $2^{-(n-1)}$ probability that  $\voters_j = \emptyset$ for some $j \in \onetwo$. In such a case, we take $f(\profile^{(j)})$ to be the empty ranking (all alternatives tied) for all $f \in F$. This effectively gives zero weight to all such splits when comparing expected disagreements.} the disagreement when applied separately to two sides of a random split of $\profile$. 
 We now illustrate which rule $\aba$ would pick on an example profile by showing how \eqref{eq:aba} can be computed or bounded for each candidate rule.

\begin{ex}\label{ex:aba_run}
    Fix some integer $k \geq 2$ and consider a profile $\profile$ with alternatives $\cand =\{a,b,c\}$ and $\nvoters = 3k$ voters. Suppose the voters comprise three groups, consisting of:
    \begin{itemize}
        \item Group 1: $k$ voters that rank $a \succ b \succ c$.
        \item Group 2: $k$ voters that rank $a \succ c \succ b$.
        \item Group 3: $k$ voters that rank $b \succ c \succ a$.
    \end{itemize}
    Consider the set of candidate rules $F=\{f_p, f_v\}$ (\emph{i.e.}, plurality and veto). As defined in \Cref{sec:prelim}, plurality $(f_p)$ simply ranks alternatives in decreasing order of the number of voters that rank them top. Veto $(f_v)$, on the other hand, ranks alternatives in increasing order of the number of voters that rank them bottom. Let us now evaluate the two candidate rules, plurality and veto, under our Aggregation by Consistency ($\aba$) RPR. Say $\profile^{(1)}$ and $\profile^{(2)}$ are the random variables corresponding to the two subprofiles resulting from splitting $\profile$ via the random process in \Cref{def:aba}. First, consider plurality $f_p$. For notational convenience, let $r_p^1 = f_p(\profile^{(1)})$ and $r_p^2 = f_p(\profile^{(2)})$ denote the random variables corresponding to the two output rankings. Since alternative $c$ does not appear on top of any voter's ranking, both $r_p^1$ and $r_p^2$ will rank $a \succ c$ unless all of the $2k$ voters ranking $a$ as their top alternative end up on the same side of split, in which case the ranking of the other side will have $a$ and $c$ tied. Similarly, both $r_p^1$ and $r_p^2$ will rank $b \succ c$ unless all of the $k$ voters ranking $b$ top end up on the same side of split. Thus,
    \begin{align*}        \expc\left[D_{r_p^1,r_p^2}^{a,c}\right]=  \expc\left[D_{r_p^1,r_p^2}^{b,c}\right] = 0; \quad \expc\left[T_{r_p^1,r_p^2}^{a,c}\right] \leq \frac{1}{2^{2k-1}}; \quad  \expc\left[T_{r_p^1,r_p^2}^{b,c}\right] \leq \frac{1}{2^{k-1}}\nonumber.
    \end{align*}
    Using the linearity of expectation, this implies the expectation in (\ref{eq:aba}) for $f_p$ is
    \begin{align}
        \expc\left[KT\left(r_p^1,r_p^2\right)\right] & = \expc \left[ D_{r^1_p,r_p^2}^{a,b} + \frac{1}{2}T^{a,b}_{r^1_p,r_p^2}\right] + \expc \left[ D_{r^1_p,r_p^2}^{a,c} + \frac{1}{2}T^{a,c}_{r^1_p,r_p^2}\right]+\expc \left[ D_{r^1_p,r_p^2}^{b,c} + \frac{1}{2}T^{b,c}_{r^1_p,r_p^2}\right] \nonumber \\ & \leq 1 + \left( 0+ \frac{1}{2} \cdot  \frac{1}{2^{2k-1}}\right)+\left( 0+ \frac{1}{2} \cdot  \frac{1}{2^{k-1}}\right) \leq 1+\frac{1}{2^4}+\frac{1}{2^2}=1.3125.\label{eq:exp}
    \end{align}
     Now, consider veto $f_v$. Say $r_v^1 = f_v(\profile^{(1)})$ and $r_v^2 = f_v(\profile^{(2)})$. Since the number of voters that rank $a$ and $b$ as their bottom alternative in $\profile$ are tied, $r_v^1$ and $r_v^2$ will either disagree about how $a$ and $b$ should be ranked, or both of them will tie $a$ and $b$. In other words, we cannot have $ D_{r^1_v,r_v^2}^{a,b}= T_{r^1_v,r_v^2}^{a,b}=0$. The same is true for any other pair of alternatives. This implies the expectation in (\ref{eq:aba}) for $f_v$ is
    \begin{align}
        \expc\left[KT\left(r_v^1,r_v^2\right)\right] & = \expc \left[ D_{r^1_v,r_v^2}^{a,b} + \frac{1}{2}T^{a,b}_{r^1_v,r_v^2}\right] + \expc \left[ D_{r^1_v,r_v^2}^{a,c} + \frac{1}{2}T^{a,c}_{r^1_v,r_v^2}\right]+\expc \left[ D_{r^1_v,r_v^2}^{b,c} + \frac{1}{2}T^{b,c}_{r^1_v,r_v^2}\right]   \geq 1.5. \label{eq:exv}
    \end{align}
    Comparing \eqref{eq:exp} and \eqref{eq:exv}, we see that $\aba(F, \profile)=\{f_p\}$. Choosing plurality over veto indeed makes sense for aggregating the rankings in $\profile$, as the plurality (first place) scores of each alternative in $\profile$ clearly give more information about how they compare to each other than their veto (last place) scores. 
\end{ex}

We now discuss three extensions of $\aba$, which we later implement (\Cref{sec:experiments}).

\paragraph{Approximating expected disagreement} While computing the expected disagreement in \eqref{eq:aba} for each SWF can be difficult for more complicated profiles and candidate rules than in \Cref{ex:aba_run}, in practice $\aba$ can be approximated for finite number of candidate rules $|F|$ via Monte Carlo sampling, \emph{i.e.}, by splitting voters via the random process in \Cref{def:aba}, computing the disagreement $KT(f(\profile^{(1)}),f(\profile^{(2)}))$ for each SWF $f \in F$, and repeating the process for a desired number of splits, eventually returning the SWF with the minimum average disagreement.

\paragraph{Infinitely many candidate rules} The sampling approach above works even for certain candidate rule sets $\rules$ with infinite $|F|$. In such cases, optimization algorithms can be used to find the minimizer of (\ref{eq:aba}). For example, if the set of candidate rules is the set of all positional scoring rules (\Cref{def:pos}), \emph{i.e.}, $F=F_S$, one can run a constrained optimization program with the scoring vector $(s_j)_{j \in \{1,2,\ldots,\ncand\}}$ as the variables. Indeed, our implementation of $\aba$ for the experiments in \Cref{sec:experiments} relies on Monte Carlo sampling and includes (among other rules) an optimization over all (infinitely many) positional scoring rules.

\paragraph{Partial rankings} A more general setting than the one we have considered so far is that of \emph{partial rankings}, where each voter $i \in \voters$ ranks a subset of alternatives $\cand_i \subseteq \cand$. This setup is more appropriate for certain settings in which we would like to utilize $\aba$, such as peer review. To extend $\aba$ to partial rankings, we must take into account that each alternative is no longer ranked by every voter. As a result, some alternatives will naturally have more evaluations in a single side of a random split, while others' evaluations will be more evenly split. Ideally, disagreements over alternatives that are well-represented on both sides should be penalized more harshly than those concerning alternatives that appear primarily on one side, where the other side may lack sufficient information. Indeed, in the extreme case where all evaluations of an alternative end up on the same side of a random split, it seems unreasonable to expect any SWF to place this alternative consistently across the split. This can be achieved by replacing the $KT$ function in (\ref{eq:aba}) with the \emph{weighted $KT$ function}~\citep{Kumar10:Generalized}, where the weight of each alternative is set according to how evenly they are represented across a split. Indeed, our implementation that was awarded in the IJCAI-24 competition mentioned in \Cref{sec:highlevel} used this method, as well as some of our experiments in \Cref{sec:experiments}.

\paragraph{Score data} Some elicitation systems may require voters to submit scores for the alternatives rather than ranking them. In such settings, one option is to convert each voter's scores into a ranking, and then apply $\aba$ to pick a rank aggregation rule. Alternatively, we can also extend $\aba$ to directly pick among aggregation functions that take scores as input rather than rankings (\emph{e.g.}, maximum, geometric/arithmetic mean). To do so, we split the scores received by each alternative via a process analogous to that in \Cref{def:aba}, once again choosing the aggregation method that outputs most consistent results across the split. Indeed, one of our experiments in \Cref{sec:experiments} implement this extension.

\section{Axiomatic Analysis of \emph{AbC}}\label{sec:axiom}
Having introduced our rule picking rule $\aba$, we would like to investigate its axiomatic properties. To do so, we now introduce some natural axioms for RPRs, focusing on the setting of full rankings introduced in \Cref{sec:prelim}. While some of these axioms are inspired by their counterparts for SWFs, not all SWF axioms can be easily translated to our RPR framework. After all, while an SWF maps preferences over alternatives to an aggregate ranking of alternatives (\emph{i.e.}, its inputs and output are rankings of the same set), an RPR must pick SWFs from candidate rules based on preferences over alternatives, not over candidate rules.

\subsection{Consistency axioms}\label{sec:consistency}

Our primary goal in employing RPRs is to effectively identify voting rules that are well-suited to a given profile. It thus becomes highly desirable for these RPRs to respond in a predictable and consistent manner when the profile changes. To rigorously analyze this desired consistency, we introduce axioms that formalize the expected behaviors of RPRs in response to changes in the profile. This approach adapts established social choice theory axioms, originally designed for SWFs, for application to RPRs, as well as introducing new axioms specific to our setting.

\subsubsection*{\underline{Reversal symmetry} }
One basic change to a profile is to flip the ranking of every voter, so their formerly top ranked alternative is now ranked bottom, and so on. In such a situation, we would ideally like the SWF(s) picked by our RPR to flip too, due the symmetry between the two directions of reading a ranking (top to bottom versus bottom to top). More formally, given a positional scoring rule $f_s \in F_S$, we define the reverse of $f_s$, denoted $\rev(f_s)$, as the scoring rule $f_{s'}$ associated with vector $s'=(s'_1,s'_2, \ldots, s'_{\ncand})=(1-s_\ncand,1-s_{\ncand-1}, \ldots, 1-s_{1})$. For instance, the reverse of plurality is veto, whereas Borda count is its own reverse. For a set of rules $F \subseteq F_S$, we write $\rev(F)=\{\rev(f): f \in F\}$. For any (strict or weak) ranking $r \in \tL(\cand) \cup \wL(\cand)$, we say $\rev(r)$ is the reversed ranking (\emph{e.g.}, $a \succ b \succ c$ becomes $c \succ b \succ a$). Lastly, we use $\rev(\profile)= \{\rev(\sigma_i)\}_{i \in N}$ to denote the reverse of a profile $\profile$. With this notation in place, we now define reversal symmetry,\footnote{Our definition for reversal symmetry is inspired by the homonymous property for rank aggregation methods introduced by \citet{Saari94:Geometry}. In that context, reversal symmetry dictates that if we reverse every voter's ranking, then the output of the method should also be the reverse of its original output. \citet{Holliday23:Split} adapt this axiom to social choice functions (which, given a profile, return a subset of alternatives); here, reversal symmetry requires that if an alternative is the unique winner of a profile, it should not be a winner of the reversed profile. Our axiom for RPRs is more demanding in the sense that we not only require the picked positional scoring rules to change, but also that they should be replaced exactly by their reverse rules.

} our first axiom for RPRs.

\begin{defn}[Reversal symmetry]\label{def:revsym}
    An RPR $\rpr$ satisfies \emph{reversal symmetry} if for any subset of positional scoring rules $F \subseteq F_S$ such that $\rev(F)=F$ and for all profiles $\profile$, we have $\rev(\rpr(\rules, \profile))=  \rpr (\rules, \rev(\profile))$. 
\end{defn}
To see why this is a natural property, suppose for example the ``signal'' in the votes is concentrated at the top: for each voter, the choice of the top alternative is statistically informative but the remainder of the ranking is not.  In such a case, a reasonable RPR should choose plurality.  But if we flip all the votes, then the signal is concentrated at the bottom, and a reasonable RPR should choose veto. As shown in \Cref{thm:axiombundle} below, $\aba$ satisfies this axiom, which follows from the symmetries of the Kendall-Tau distance function.

\subsubsection*{\underline{Shuffling consistency}} 
As with the example above, in some profiles certain parts of the rankings may be more informative than others. To axiomatize the behavior of RPRs in such situations, we will now introduce the novel concept of \emph{shuffling} a profile. Intuitively, we would like to capture the idea that if we uniformly ``shuffle'' an interval of positions in all voters' rankings, exactly where in this interval an alternative ends up reveals no information about how it compares to other alternatives in the interval. Thus, for such a ``shuffled profile,'' we would want our RPR to pick rules that treat these positions equivalently. For example, if in every voter's ranking we uniformly at random permute all of the alternatives they ranked 2$^\text{nd}$ to $|\cand|^{\text{th}}$, whether an alternative is ranked 2$^\text{nd}$ or 3$^\text{rd}$ by any voter no longer makes a difference. Given such a profile, an RPR that, for instance, picks from positional scoring rules should ideally pick one that gives positions 2 and 3 roughly equal weight.

We now formally define shuffling. Intuitively, to shuffle a given subset of $t \leq \ncand$ indices in every voter's ranking, we would like to create a new profile where all $t!$ permutations of these indices are equally represented. To do, we will eventually split the profile uniformly between these permutations. Since the number of voters may not be divisible by $t!$, we will first create a multiple of $\ncand!$ copies of the profile, which ensures that we can assign an equal number of copies for each of the $t!$ permutations.

\begin{defn}[Shuffled profile]\label{def:shuffle}
    Given a profile $\profile$ and a subset of indices $S \subseteq \{1,2,\ldots,\ncand\}$, the \emph{$k$-shuffling of $\profile$ with respect to $S$} (denoted $\shuffle{S}{k}$) is a profile obtained as follows: For each voter $i \in \voters$,
    \begin{itemize}
        \item Create $k\cdot \ncand!$ copies of voter $i$'s ranking $\sigma_i$ and separate them into $|S|!$ groups of equal size, each assigned to a unique permutation of the indices in $S$.
        \item Modify the votes in each group so that the alternatives ranked in the positions in $S$ are permuted according to the permutation assigned to that group. 
        \item Add all copies to the final profile.
    \end{itemize}
\end{defn}
For example, say $\profile$ consists of a single ranking $a \succ b \succ c$, and we are interested in its 1-shuffling with respect to indices $\{1,2\}$. By the first step of \Cref{def:shuffle}, we first create $k\cdot \ncand! =1 \cdot 3! = 6$ copies of this ranking. We then split them into $|S|!=2!=2$ groups of size 3 each. The first group is assigned to the identity permutation of the first two positions. The second group is assigned to the permutation that reverses the first two positions. Hence, $\shuffle{\onetwo}{1}$ consists of six rankings, with three of them being $a \succ b \succ c$ and three being $b \succ a \succ c$.

Our next axiom dictates an RPR's behavior in the extreme case of shuffling all but one position.

\begin{restatable}[Plurality-shuffling consistency]{defn}{pscdef}\label{def:psc}
    Say we are given a finite set of positional scoring rules $\rules \subset \rules_S$ that  contains plurality (\emph{i.e.}, $f_p \in F$) and a profile $\profile$ such that $f_p(\profile)$ contains no ties. Then, an RPR $\rpr$ satisfies \emph{plurality-shuffling consistency} if for any such $F$ and $\profile$, there is a $k\in \mathbb{Z}_+$ so that $\rpr(\rules, \shuffle{[2,\ncand]}{k'})=\{f_{p}\}$ for all $k'\geq k$, \emph{i.e.}, $\rpr$ picks \emph{only} $f_p$ for the $k'$-shuffling with respect to $[2,\ncand]=\{2,3,\ldots,\ncand\}$.
\end{restatable}

Intuitively, for a (sufficiently large) profile where the top position of the votes gives an unambiguous ranking (under plurality), but the remaining $m-1$ positions are effectively indistinguishable, an RPR satisfying plurality-shuffling consistency identifies the rule that treats these $m-1$ positions equally (\emph{i.e.}, plurality) as the only appropriate rule. Such a profile, for example, can approximate settings where the voters clearly know their top alternative, but are indifferent for the remaining spots. Plurality-shuffling consistency is a natural property that respects the symmetry in the profile and identifies where the ``signal'' in the votes is concentrated. Despite this, a large class of RPRs, which we introduce next, all fail plurality-shuffling consistency, as we will show in \Cref{thm:axiombundle}.

\begin{defn}[Welfare-maximizing RPRs]\label{def:welfaremax}
    An RPR $\rpr$ is \emph{welfare-maximizing} if there exists a utility function $u: \tL(\cand) \times \wL(\cand) \rightarrow \mathbb{R}$ such that $\rpr(\rules, \profile)=\argmax_{f \in F} \sum_{i=1}^n u(\sigma_i, f(\profile))$ for all profiles $\profile$ and candidate rule sets $\rules$. 
\end{defn}

Welfare-maximizing RPRs capture an approach common in prior work: interpreting the optimal voting rule as one that  maximizes social welfare with respect to some utility function~\citep{Caragiannis11:Voting,Gershkov17:Optimal}. Indeed, this is how voting rules in the IJCAI-24 competition mentioned in \Cref{sec:highlevel} were scored. As \Cref{def:welfaremax} puts no restrictions on the utility function $u$, welfare-maximizing RPRs constitute a wide class. As such, whether a specific member of this class satisfies an axiom may depend on its associated utility function. For example, it is relatively straightforward to show that a welfare-maximizing RPR satisfies reversal symmetry if its utility function is reversal symmetric, \emph{i.e.}, $u(r,r')=u(\rev(r),\rev(r'))$ for all $r \in \tL(\cand)$ and $r' \in \wL(\cand)$. Nevertheless, we show in \Cref{thm:axiombundle} that \emph{all} welfare-maximizing RPRs fail plurality-shuffling consistency.

On the other hand, this is not the case for $\aba$, which indeed satisfies the axiom. To see this, take a random split of the shuffled profile $\shuffle{[2,\ncand]}{k}$. Given the balanced nature of all positions except the first, any rule that treats a distinction between alternatives in these positions as a signal will inevitably observe the reversed signal on the other side of the split, leading to a higher disagreement. We formalize this in \Cref{thm:axiombundle}.

\subsubsection*{\underline{Union consistency}} 
Next, we study the behavior of RPRs when we combine sets of voters. The next axiom we introduce, \emph{union consistency}, is inspired by an analogous axiom for voting rules, simply named consistency~\citep{Young75:Social}. Informally, it dictates that whenever an alternative is the winner of a voting rule when applied separately to the rankings of two distinct sets of voters, the same alternative should still win when we bring those sets of voters together. We now define our version of the axiom, specifically for RPRs.

For any two profiles $\profile_a$ and $\profile_b$ over the same alternatives $\cand$ but with two disjoint sets of voters $N_a$ and $N_b$, respectively, let $\profile_a + \profile_b$ denote the profile with the rankings of all voters in $N_a \sqcup N_b$. 

\begin{defn}[Union consistency] An RPR $\rpr$ satisfies \emph{union consistency} if for any candidate rule set $F$ and profiles $\profile_a$ and $\profile_b$ such that $\rpr(\rules, \profile_a) \cap \rpr(\rules, \profile_b) \neq \emptyset$, we have
\begin{align*}
    \rpr(\rules, \profile_a+\profile_b) = \rpr(\rules, \profile_a) \cap \rpr(\rules, \profile_b).
\end{align*}
\end{defn}

Unlike previous axioms, $\aba$ does not satisfy UC, but (as we show below) neither does any RPR satisfying reversal symmetry and plurality-shuffling consistency, giving us our first impossibility result for RPRs. Further, we argue that union consistency is less significant for an RPR than its counterpart for voting rules: just because a rule is appropriate for two different sets of rankings does not necessarily mean it is appropriate for their union, especially if they are differently distributed. Indeed, in the extreme case of a single voter, \emph{any} sensible aggregation rule (\emph{i.e.} that returns the ranking itself) is appropriate.

\subsubsection*{\underline{Main result of \Cref{sec:consistency}}} 

We now present \Cref{thm:axiombundle}, which states our results regarding the consistency axioms in \Cref{sec:consistency}. 

\begin{restatable}{thm}{axiombundle}\label{thm:axiombundle}
    (1) $\aba$ satisfies reversal symmetry.\\ (2)     Any welfare-maximizing RPR fails plurality-shuffling consistency; $\aba$ satisfies it.\\ (3) No (anonymous) RPR can satisfy all three of reversal symmetry, plurality-shuffling consistency, and union consistency.
 \end{restatable}

The proof can be found in \Cref{app:consistency}. The axioms in \Cref{thm:axiombundle} are specific for RPRs, demonstrating $\aba$ responds consistently to changes in the profile. Next, we will turn our attention to existing axioms for SWFs and how $\aba$ (and RPRs in general) interact with them.

\subsection{Preserved axioms}\label{sec:preserve}

In this section, we investigate which axiomatic properties (for SWFs) are retained by an RPR when all SWFs in the set of candidate rules satisfy them.

While using an RPR has advantages beyond just outputting an aggregate ranking (see \Cref{sec:intro}), it is true that an RPR $\rpr$, paired with a set of candidate rules $\rules$, \emph{induces} an SWF:
\begin{align*}
    f^\rules_\rpr(\profile) \eqdef f(\profile) \text{ where }\rpr(\rules, \profile)=\{f\},
\end{align*}
\emph{i.e.}, $f_\rpr^\rules$ is the SWF that first maps an input profile to an SWF using $\rpr$, and then applies that SWF to the profile. For $f^\rules_\rpr$ to be well-defined, we must have $|\rpr(\rules, \profile)|=1$. As we are interested in the properties of $f^\rules_{\aba}$, for this subsection alone, we restrict our setting to candidate rules and profiles for which $|\aba(\rules, \profile)|=1$, and to RPRs $\rpr$ for which $|\rpr(\rules, \profile)|=1$ for these profiles.

Clearly, the axiomatic properties of $f_{\rpr}^\rules$ as an SWF for any RPR $\rpr$ depend on our selection of the candidate rules $F$. By restricting $\rules$ to SWFs that all satisfy a certain property, can we ensure that so will $f_\rpr^\rules$? More formally, we say an RPR $\rpr$ \emph{preserves} a property $P$ if whenever $P$ is true for all candidate rules $f \in \rules$, then $P$ is true for $f^\rules_\rpr$. For example, recall that an SWF is anonymous if permuting the voters $\voters$ in a profile does not change the output ranking of the SWF, \emph{i.e.}, all voters are treated equally. Similarly, an SWF is neutral if permuting the alternatives $\cand$ in a profile results in the output ranking of the SWF being permuted in the same way, \emph{i.e.}, all alternatives are treated equally. It is immediate from \Cref{def:aba} that $\aba$ preserves these two properties. On the other hand, this is not true for all RPRs,  \emph{e.g.}, for an alternative $a \in \cand$, say the RPR $\rpr_a$ maps each profile $\profile$ to the candidate rule $f \in F$ that ranks $a$ highest in $f(\profile)$. Then, $f_{\rpr_a}^\rules$ is clearly not neutral, even if all $f \in \rules$ are. Still, as we show in \Cref{thm:preservebundle} below, certain fundamental social choice axioms are preserved by \emph{all} RPRs. 

However, not every axiom is as easily preserved. Given a weak ranking $r\in \wL(\cand)$ and an alternative $a \in \cand$, say $\rank_r(a)=1+ |\{b \in \cand: b \succ_r a \}|$ is the rank of $a$ in $r$. We say an SWF $f$ satisfies \emph{monotonicity} if for any profile $\profile$ and alternative $a \in \cand$, if $\profile'$ is the same as $\profile$ except some voters now rank $a$ higher, then $\rank_{f(\profile)}(a) \geq \rank_{f(\profile')}(a)$. In words, monotonicity dictates that promoting an alterative in some voters' rankings while keeping all else constant should not hurt the alternative in the final output ranking. It turns out $\aba$ does not preserve monotonicity. This is  because promoting an alternative may still cause the rule picked by $\aba$ to change to a SWF that ranks that alternative further below, even if it does not hurt that alternative under any fixed candidate rule's output ranking. Still, much like union consistency, we show that preserving monotonicity is not compatible with another axiom satisfied by $\aba$, reversal symmetry.

We now present \Cref{thm:preservebundle}, which states our main results from \Cref{sec:preserve}.

\begin{restatable}{thm}{preservebundle}\label{thm:preservebundle}
    (1) $\aba$ preserves anonymity and neutrality.
    
    \noindent(2) Any RPR $Z$ preserves the Smith criterion, Condorcet consistency, majority winner, pairwise majority consistency, and unanimity.
    
    \noindent(3) No (anonymous) RPR can satisfy reversal symmetry and preserve monotonicity.
\end{restatable}

For readers unfamiliar with the axioms in (2) of \Cref{thm:preservebundle}, we provide the formal definitions in \Cref{app:rpr_preserve}, along with the proof of the full theorem. Statement (2) in the theorem demonstrates a strength of not just $\aba$, but of our RPR framework in general, showing how one can still reap the benefits of the axiomatic approach simply by restricting the set of candidate rules $\rules$ to certain SWFs.

Overall, we have shown that $\aba$ satisfies many natural axioms (reversal symmetry, plurality-shuffling consistency, and preservation of the axioms in \Cref{thm:preservebundle}); as for axioms it fails (union consistency and the preservation of monotonicity), we have given impossibility results showing each of them is incompatible with axioms that $\aba$ satisfies. In \Cref{app:axiom_violation}, we apply the data-driven approach of \citet{caiata2025} to show experimentally that violations of monotonicity and union consistency are typically rare across well-known distributions, as well as on real-life data.


\section{Computational Problem: $\perfpos$}\label{sec:computational}

Recall from \Cref{def:aba} that $\aba$ picks the SWF(s) from candidate rule set $\rules$ that minimizes the expected disagreement over a random split. In this section, we show that when $F=F_S$ (candidate rules are the set of all positional scoring rules), minimizing the disagreement over even a given split is algorithmically hard. Consider the following computational problem: We are given a strict ranking $\sigma_i \in \tL(A)$ for each voter $i \in \voters$, as well as an even split of voters $\voters=\voters_1 \sqcup \voters_2$ with $|N_1|=|N_2|$. Then, $\perfpos$ (\textbf{Perf}ect \textbf{Pos}itional) asks:

\begin{center}
\begin{minipage}{11cm}     Is there a positional scoring rule $f_s \in F_S$ that achieves perfect consistency over this split, \emph{i.e.}, obtains $KT(f_s(\profile^{(1)}),f_s(\profile^{(2)}))=0$?
\end{minipage}
\end{center}
Here, for each $ j\in \onetwo$, $\profile^{(j)} = (\sigma_i)_{i \in N_j}$ is the restriction of the profile $\profile$ to voters in $N_j$. We now show that this problem is hard.
\begin{restatable}{thm}{nphard}\label{thm:perfpos}
    $\perfpos$ is \NP-complete.
\end{restatable}
The proof, which is a reduction from 3SAT, can be found in \Cref{app:perfpos}. \Cref{thm:perfpos} shows that when $F=F_S$,
it is even hard to decide whether it is possible to have no disagreement at all. This also proves that the problem of computing the minimal possible disagreement for a given split is not only hard, but also hard to approximate to any multiplicative factor (since any algorithm with such a guarantee must return a solution with disagreement 0 when possible). While technically distinct, our result is aligned with other hardness results for optimizing over positional scoring rules, \emph{e.g.}, for picking the rule most consistent with an underlying true ranking to which we are given partial access~\citep{Caragiannis19:Optimizing}.

It is worth noting that \Cref{thm:perfpos} does not rule out efficient implementations of $\aba$ with candidate rule sets $\rules$ that do not contain all positional scoring rules. In particular, if $\rules$ is finite and only contains efficiently implementable rules, minimizing disagreement for a given split is easy, since the disagreement of each SWF can be computed one by one. We leave studying the complexity of computing $\aba$ with different sets of candidate rules for future work. Further, despite the complexity result, as we will experimentally show in the next section, an implementation of an approximation of $\aba$ using Monte Carlo sampling (see the discussion of extensions to $\aba$ in \Cref{sec:aba}) is computationally efficient in practice and yields desirable results, even when $F_S \subseteq F$ (all positional scoring rules are included in the candidate rule set).

\section{Experiments}\label{sec:experiments}

We now give experimental results evaluating $\aba$ and discuss their implications. Given a profile $\profile$, in order to approximate the expectation in Equation~\eqref{eq:aba}, our implementation samples a number of random splits of $\profile$ according to the procedure in \Cref{def:aba}, computing the disagreement for each candidate rule on each split, and returning the rule with the minimum average disagreement. We first describe procedures and parameters common through all our experiments. Additional details regarding the experiments, including specific datasets, distributions we use, and parameter settings can be found in \Cref{sec:appendix_experiments}. 
\begin{itemize}
    \item \emph{Number of Splits}: With two exceptions, all experiments report the average disagreement resulting from each rule over 10 random splits of voters. The exceptions are: When evaluating axioms (\Cref{app:axiom_violation}, \autoref{fig:axiom_violation_rates_appendix}) we report the average over 50 splits, and when evaluating score data (\Cref{sec:score_data}, \autoref{tab:score_data}) we report the average over 1000 splits. These values were decided based on computational and time constraints.
    \item \emph{Splitting Procedure}: When creating two groups of voters we assign each voter to one group or the other uniformly at random, as detailed in \Cref{def:aba}. In the case that no voters are assigned to one group, we assign that group a single weak ranking in which all alternatives are tied (see Footnote~\ref{fn:emptysplit}).
    \item \emph{Weighting Partial Rankings}: 
    In settings where voters submit partial rankings, we weigh each alternative based on the number of times it appears on each side of a split, as discussed in \Cref{sec:aba}. This serves to correct for mismatches in the amount of information known about each alternative (\emph{i.e.}, the number of times each alternative is ranked) between the two sides of the split. Consider a split of a fixed profile $\profile$ into ($\profile^{(1)}, \profile^{(2)})$ via the process in \Cref{def:aba}. For any alternative $a \in \cand$, let $m_a$ denote the minimum number of times $a$ is included in rankings in $\profile^{(1)}$ or $\profile^{(2)}$, and let $t_a$ refer to the total number of times $a$ is included in a ranking across the entire profile $\profile$. Then the weight of $a$ is $w_a = \frac{\gamma^{m_a}-1}{\gamma^{t_a/2}-1}$. This ensures that the largest weights are assigned to alternatives that are evenly split across the split ($m_a=t_a/2$), whereas any alternative that does not appear on one side of the split ($m_a=0$) gets zero weight, as it is unreasonable to expect this alternative to be placed consistently across the split by any SWF. All our experiments use a constant factor of $\gamma = 2$. Weights are used as a scaling factor in computing the Kendall-Tau distance (as in \citealt{Kumar10:Generalized}); when alternatives $a$ and $b$ are in opposite positions across two rankings then they contribute $w_aw_b$ to the weighted Kendall-Tau distance (as opposed to the unweighted setting where any misordered pair adds 1 to the distance).
    \item \emph{Computational Power}: All experiments are performed locally on a 2022 M2 Macbook Air with 16 GB of memory. Each individual experiment shown in the paper took between 0.5 and 3 hours to complete.
\end{itemize}

As discussed in \Cref{sec:aba}, the SWFs we test include an optimization over all positional scoring rules (``Best Positional Scores''), which uses simulated annealing to search for the scoring vector minimizing the average disagreement for a given profile (see \Cref{app:bestpos} for details). ``Trimmed Borda'' refers to a modified version of Borda count where the maximum and minimum rank received by each alternative is deleted from each split, before computing its Borda score with the remaining evaluations \citep{Meyer22:Analysis,Stafford25:RoRI}. All reported distances are normalized by dividing the (weighted) Kendall-Tau distance resulting from each rule to the maximum (weighted) Kendall-Tau distance obtainable for that split (\emph{i.e.}, the distance between two reversed rankings). We now explain the setup of each experiment separately, along with a discussion of their results.

\begin{figure}[t]
    \centering
    \includegraphics[scale=0.43]{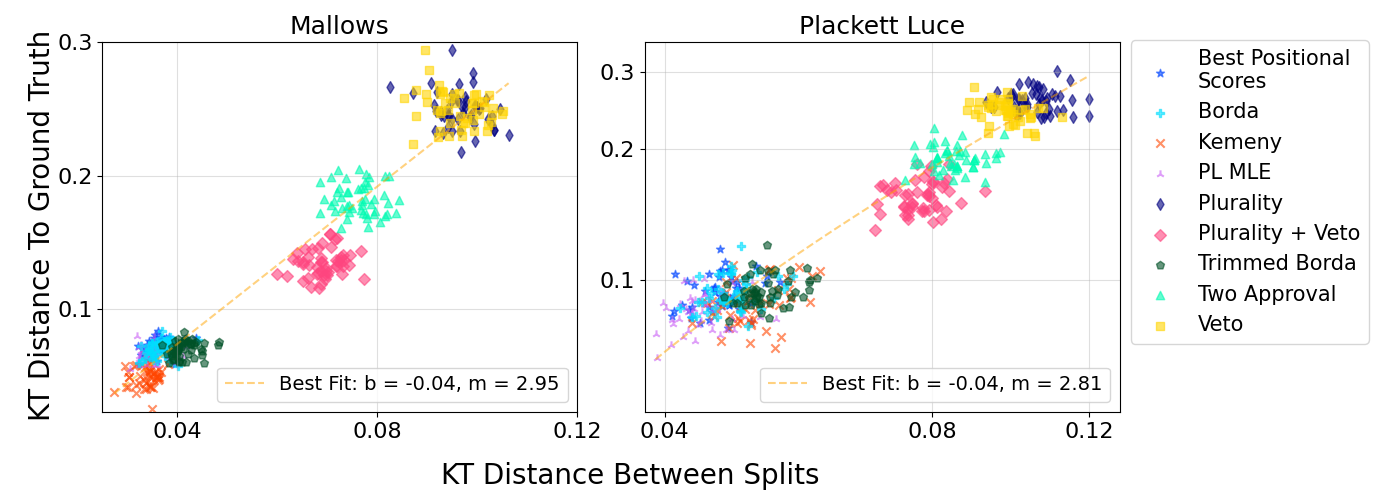}
    \caption{
    Log-log plots of $KT$ distance between the ground truth ranking and the ranking generated by SWFs on the complete profile vs.\ $KT$ distance between rankings generated by SWFs on splits of the profile. The profiles consist of partial rankings drawn from the Mallows and Plackett-Luce (PL) models. Each point shows average distance over 10 splits for a single profile. In both cases, the MLE of the model outperforms other SWFs in both axes. 
    }
    \label{fig:ground_truth_vs_central_ranking}
\end{figure}

\subsection{Ground truth distance vs.\ disagreement}  \label{sec:gt_v_disagreement}

While $\aba$ does \emph{not} assume a ground truth, what happens if a generative model with a ground truth is indeed a reasonable approximation of the data? In order to investigate this question, we apply $\aba$ to partial rankings synthetically drawn from two well-known distributions: Mallows~[\citeyear{Mallows57:Non}] and Plackett-Luce~\citep{Plackett75:Analysis,Luce59:Individual}. We use profiles of $100$ voters and $100$ items; each voter ranks $10$ items and each item is ranked $10$ times. Since each distribution has a known ground truth, for each SWF in the candidate rule set, we can measure the distance between this ground truth ranking and the ranking generated by the SWF when applied to the entire data set. We then calculate each SWF's disagreement over random splits, following the procedure in \Cref{def:aba}. 
Since $\aba$ picks the SWF that minimizes this disagreement, we would ideally like to observe that the rule that minimizes the distance between splits (most consistent rule) corresponds to the rule which minimizes the distance between its output on the full profile and the ground truth (the rule with the minimum error).

Indeed, this is the case: 
\autoref{fig:ground_truth_vs_central_ranking} highlights a clear positive relationship between distance to ground truth and disagreement between splits, confirming our intuition from MVUEs (\Cref{sec:highlevel}). Importantly, for both Mallows and Placket-Luce, the MLE of the model (Kemeny and PL MLE, respectively) outperforms the other SWFs in both axes. As each MLE has the lowest disagreement for its model, $\aba$ would indeed pick it when given data from this model, hence obtaining the benefits of the statistical approach.

\subsection{Scientific peer review (score data)}
\label{sec:score_data}

In some cases evaluators submit scores rather than rankings. As discussed in the extensions listed in \Cref{sec:aba}, $\aba$ is still applicable in this case. Rather than considering functions for aggregating sets of preference rankings, we can also use our method to consider functions for aggregating sets of review scores, in order to pick the one yielding the most consistent results.
We explore this setting on a dataset consisting of peer review scores given to proposals for observatory telescope time \citep{Kerzendorf20:Distributed}. In this data, 136 reviewers have provided feedback on between 4 and 8 proposals each. We filter data to consider the 119 proposals that received at least 6 scores (No proposal was reviewed by more than 8 reviewers).

Rather than splitting the set of reviewers to two sides as done in previous sections, we individually split the 6 to 8 scores received by each unique proposal. Each split is generated by randomly placing half of the scores into one side of the split, and the remaining half into the other side. Where there are an odd number of scores we select a random score to leave out of either side, resulting in each proposal having a split into two sides of equal size.
We then run each score aggregation function we consider (listed below) on each side of the split for every proposal, resulting two sets of aggregate scores. From these we generate two rankings over all proposals, one for each side of the split. We measure the Kendall-Tau distance between these two rankings. This process (creating splits, finding an aggregate score for each split, converting scores to a ranking, measuring Kendall-Tau distance) is repeated over 1000 randomized trials to generate 1000 distance scores, and we report the average score for each function. We consider the following aggregation functions: Arithmetic Mean, Min, Max, Median, Geometric Mean, and Trimmed Mean (which takes the arithmetic mean of scores after trimming the highest and lowest score from each split).

\begin{table*}[t]
\centering
\begin{tabular}{@{}ccccccc@{}}
\toprule
Arithmetic Mean              & Min               & Max               & Median            & Geometric Mean & Trimmed Mean \\ \midrule
$0.364 \pm 0.001$ & $0.444 \pm 0.001$ & $0.409 \pm 0.001$ & $0.371 \pm 0.001$ & $0.369 \pm 0.001$ & $0.371 \pm 0.001$
\\ \bottomrule
\end{tabular}
\caption{Mean (and standard error of the mean) of $KT$ distances over 1000 splits resulting from several aggregation functions on review scores from \citet{Kerzendorf20:Distributed}.}
\label{tab:score_data}
\end{table*}

Our results, shown in \autoref{tab:score_data}, can help guide best practices for aggregating such review scores. For example, past work has suggested that items with a single high score are often prioritized in peer review~\citep{Nierstrasz00:Identify} -- corresponding to rankings generated by the Max function. In contrast, our results demonstrate that aggregation using Max scores leads to high levels of disagreement compared to functions that directly incorporate the score of multiple reviewers, such as Arithmetic/Geometric Mean.

\subsection{Scientific peer review (rank data)}\label{sec:alma_rank}

\begin{figure}[t]
    \centering
    \begin{subfigure}[t]{0.49\textwidth}
        \centering
        \includegraphics[width=\textwidth]{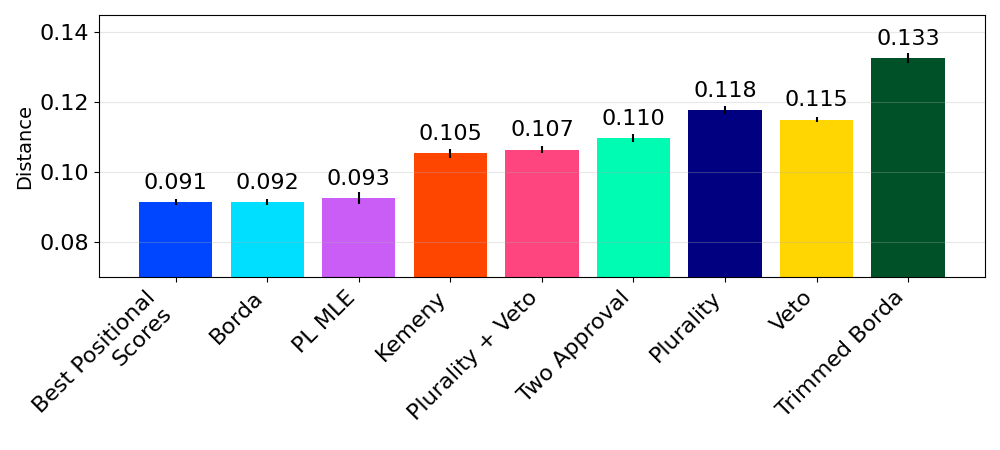}
        \caption{
        Distance between splits of partial rankings for each SWF on reviews of ALMA Cycle 10 project proposals.
        } 
        \label{fig:alma_cycle10}
    \end{subfigure}
    \hfill
    \begin{subfigure}[t]{0.49\textwidth}
        \centering
        \includegraphics[width=\textwidth]{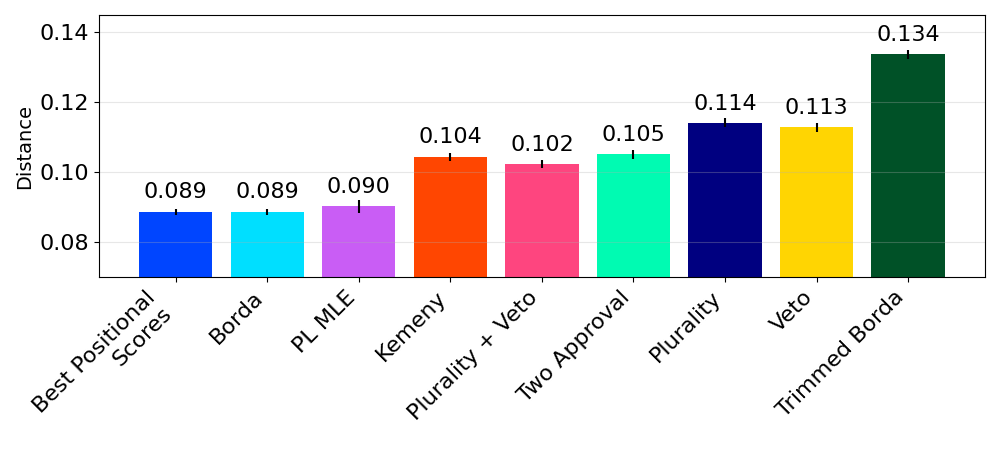}
        \caption{
        Distance between splits of partial rankings for each SWF on reviews of ALMA Cycle 11 project proposals.
        }
        \label{fig:alma_output}
    \end{subfigure}
    \caption{Split distance and standard error of several rules on partial rankings over project proposals for the Atacama Large Millimeter Array (ALMA).
    } 
    \label{fig:alma_bar_plots}
\end{figure}

Peer review data can also take the form of rankings. Here we consider anonymized rankings provided by evaluators of projects proposed for the Atacama Large Millimeter Array (ALMA) \citep{Meyer22:Analysis}, which is the largest ground-based radio telescope on earth. 

We conduct our experiments on two sets of  proposals, from Cycle 10 and Cycle 11. In this review process, each proposer is asked to review. Cycle 10 contains 1635 proposals and 1635 reviewers, while Cycle 11 contains 1729 proposals and 1729 reviewers. In each cycle, every reviewer is assigned 10 proposals to review, and every proposal is assigned to 10 reviewers. Each reviewer provides a strict ranking of the 10 proposals assigned to them. 

We run $\aba$ on the rankings provided by the reviewers. The mean Kendall-Tau distances over 10 splits for several SWFs are displayed in \Cref{fig:alma_bar_plots}. We see a pattern consistent with the results of \Cref{fig:ground_truth_vs_central_ranking}: PL MLE, Borda, and optimized positional scores result in very similar distances, suggesting a natural lower bound on the Kendall-Tau distance based on the data and the noise induced by generating splits. Kemeny rule,\footnote{Note that, due to computational limits, in this setting we use the best ranking found by the Kemeny function (\Cref{app:kemeny}) within a time limit of 15 minutes per split.} on the other hand, performs notably worse than this lower bound. Our experiment also finds that better consistency is achieved by positional scoring rules which provide more information about the profile (\textit{i.e.}, Two Approval and Plurality + Veto provide two ``bits'' of information, compared to one bit from Plurality and Veto, whereas Borda and the optimized positional scores provide significantly more information).

Importantly, our method can be used to evaluate proposed changes to current practices. For example, \citet{Meyer22:Analysis} consider removing the minimum \& maximum score of each alternative before computing the Borda score, termed `Trimmed Borda', 
but our experiment suggests this method increases disagreement on their dataset compared to Borda count, significantly more than it did on data generated from known noise models (\Cref{fig:ground_truth_vs_central_ranking}).

\subsection{Political elections} Similar to studying the impact of changes to peer-review practices, we can use our implementation of $\aba$ to evaluate the methods being actively used in real-world elections. Accordingly, \autoref{fig:empirical_scatter} considers running $\aba$ on several political elections using data from Preflib \citep{Mattei13:Preflib}. These elections took place in several municipalities across the United States between 2008 and 2012; all used Instant Runoff Voting (IRV) to determine a winner (see \Cref{sec:city-election} for the definition of IRV), which we implement as one of our candidate rules. These elections differ from our other data in that they represent human political preferences, and they have a highly varied number of voters (between 2477 and 262312) and candidates (between 4 and 25). See \autoref{tab:preflib_city_election_details} in \Cref{sec:appendix_experiments} for additional details on each election evaluated in this experiment. Unsurprisingly, we observe that elections with consistently low disagreement across voting methods (\emph{e.g.}, 3, 23) have few candidates, while elections in which no voting method achieves a low disagreement (\textit{e.g.}, 11, 19, 21) have many candidates. That is, there is likely to be at least some disagreement between rankings when ballots are quite long---simply because there are many more opportunities for rankings to differ. 

While IRV was used in practice during each of these elections it is, in some cases, the rule which gives the \emph{least} consistency (\emph{e.g.}, elections 1, 16, 18). In 21 of the 25 elections we show in \autoref{fig:empirical_scatter} there is some rule that achieves zero disagreement, however this rule is different between elections, confirming out intuition from \Cref{sec:intro} that different rules are appropriate for different settings. By always selecting the rule that results in the minimum disagreement, $\aba$ provides a reliable method of selecting a consistent ranking.

\subsection{Aggregating sporting contests}

\begin{figure}[t]
    \centering
    \begin{subfigure}[t]{0.49\textwidth}
        \centering
        \includegraphics[width=\textwidth]{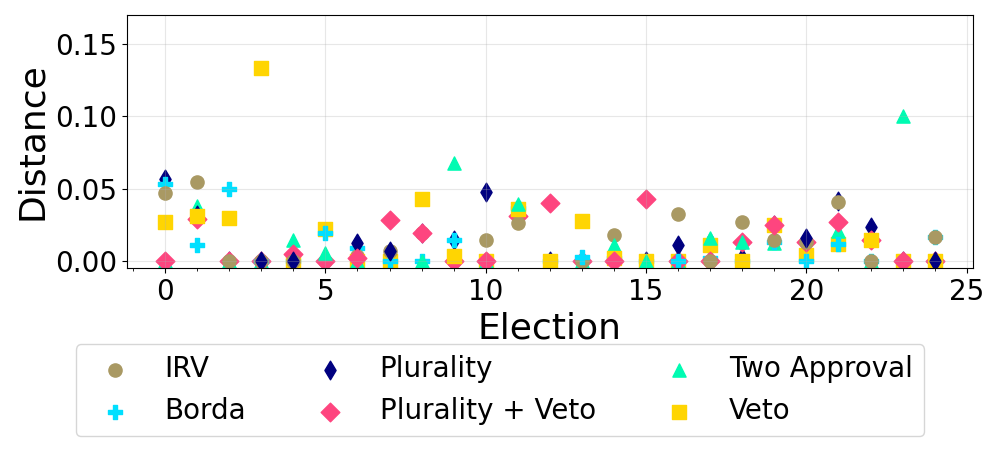}
        \caption{Distance between splits of rankings in empirical election data from cities using Instant Runoff Voting (IRV). Details on data in this plot are found in \autoref{tab:preflib_city_election_details}} 
        \label{fig:empirical_scatter}
    \end{subfigure}
    \hfill
    \begin{subfigure}[t]{0.49\textwidth}
        \centering
        \includegraphics[width=\textwidth]{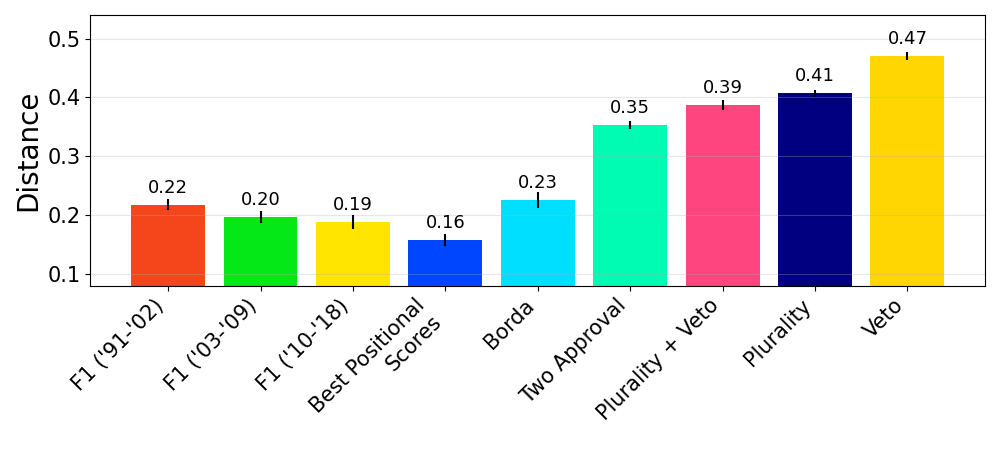}
        \caption{Distance between splits for SWFs aggregating rankings of drivers in F1 races.
        F1 vectors are evaluated only for the specific years in which that rule was in use.
        }
        \label{fig:f1_ave}
    \end{subfigure}
    \caption{Split distance with several rules on data from political elections and Formula One races.} 
    \label{fig:olympic_f1_splits}
\end{figure}

We also consider real-world settings in which participants are ranked across many events, such as sporting competitions. In these settings, each event assumes the role of a voter: the ordinal position of the competitors in each event translates directly into a ranking over competitors, which must then be aggregated into a single ranking across all events. We evaluate $\aba$ on two sporting settings: Formula One (F1) races in this section and data from the Olympics in \Cref{app:olympics}.

We use data showing F1 results for each season from 1991 to 2018, inclusive, compiled by \citet{boehmer2022quantitative}. Each season corresponds to a set of rankings (hence a profile), where each ranking records the order in which the participating drivers finished a given race. During this time period, F1 used three distinct positional scoring rules (updating the rule in 2003 and 2010) on the individual races to determine overall winners:

\begin{itemize}
    \item 1991 - 2002: $(10, 6, 4, 3, 2, 1, 0, 0, ...,0)$
    \item 2003 - 2009: $(10, 8, 6, 5, 4, 3, 2, 1, 0, 0, ...,0)$
    \item 2010 - 2018: $(25, 18, 15, 12, 10, 8, 6, 4, 2, 1, 0, 0, ...,0)$
\end{itemize}

From this data, we remove races and drivers until each race includes only drivers that competed in all events that year; we refer to the Preflib repository of \cite{Mattei13:Preflib} for details on this processing. This results in 27 years of data with a mean of 13.7 (mode of 16) races per year, and a mean of 20.6 (mode of 20) drivers per race. We then apply our standard $\aba$ methodology on this data: We generate splits for each year, apply several scoring rules to generate a ranking for each side of a split (including the three rules employed by F1), and measure the average distance between splits for each rule.

\Cref{fig:f1_ave} compares the average distances resulting from each rule. The three leftmost columns show the distances of existing F1 rules, applied \textit{only} to the years in which they were being employed. This shows that both changes to the rule being used in F1 competitions were beneficial in increasing consistency. 
The other rules shown in \Cref{fig:f1_ave} are evaluated over \textit{all} years of data. We see that recent F1 rules are significantly better than well-known positional scoring rules, including Borda count. However, our optimization of the scoring vector is able to provide some advantage over the employed rules (noting that our Best Positional Score vector is evaluated over all years while F1 rules are evaluated only on a subset of the data). These effects are demonstrated further when we evaluate all rules on each time period separately (see \Cref{app:formulaone}).

\section{Additional Related Work} 

In this section, we further elaborate on prior work conceptually related to our work. 

\paragraph{Rule selection} The question of comparing different voting rules has been widely studied in the social choice literature. For an overview of the axiomatic and computational properties of various voting rules, see~\citet{Brandt16:Handbook}. There has been earlier empirical work on collecting voters' preferences over voting rules themselves~\citep{Sertel03:Selecting,Kara05:Does,Aldrich14:Sophisticated}. However, it has been shown that (unsurprisingly) in many real-life cases voters simply prefer the rules that benefit their preferred candidate~\citep{Blais15:Citizens,Weber20:Choosing}. In general, it is clear that even social choice theorists themselves cannot reach a consensus on what is the best voting rule~\citep{Laslier12:And}. There has also been interest in using tools from automated reasoning for instantiating and generating formal justifications for using a voting rule~\citep{Cailloux16:Arguing} or imposing an axiom~\citep{Boixel20:Automated}. Similar tools were used for bypassing voting rules altogether by picking a set of axioms that impose a certain outcome~\citep{Schmidtlein23:Voting}. 

\paragraph{MLE  approach} The idea of treating votes as noisy estimates of a ground truth goes back to Condorcet~[\citeyear{CaritatMarquisdeCondorcet85:Essai}], who designed a noise model where every voter has a fixed probability of ranking each pair of alternatives correctly, and solved it for two and three alternatives. \citet{Young95:Optimal} later extended this model to an arbitrary number of alternatives, showing that its MLE is equivalent to an earlier rule introduced by \citet{Kemeny59:Mathematics}. \citet{Conitzer05:Common} later studied the question of which voting rules are MLEs for some noise model where votes are sampled independently, in particular showing that any rule that violates consistency cannot be an MLE. Combined with the result that a continuous and neutral social choice function (rules outputing a set of winners) is consistent if and only if it is a scoring rule \citep{Young75:Social}, this shows that no other (neutral and continuous) social choice function can be the MLE of any noise model. \citet{Conitzer09:Preference} did a similar analysis for social preference functions (rules outputting a set of rankings) providing an exact characterization of MLE rules as simple ranking scoring functions. Similar analyses followed for voting in multi-issue domains \citep{Xia10:Aggregating}, with partial rankings \citep{Xia11:Maximum}, and on social networks (where voters no longer vote independently) \citet{Conitzer13:Maximum}. In the opposite direction, \citet{Soufiani14:Statistical} study the MLE of two known noise models to identify whether they satisfy certain axiomatic properties, and \citet{Xia16:Bayesian} extend their results. \citet{Tideman14:Which}, on the other hand, construct a spatial noise model specifically to approximate data from actual elections, and evaluate the performance of common voting rules under this model. Lastly, a related but distinct approach to classify voting rules is distance-rationalizatability (DR) \citep{Meskanen08:Closeness,Elkind09:Distance}, which requires a rule to map each election to the result of the closest consensus according to some metric. \citet{Elkind10:Good} combine the MLE and DR approaches to better understand and compare rules, and in this framework declare Kemeny to be the best rule, as it fits both frameworks with the same underlying function. It is worth noting, however, that DR only applies to social choice functions. 

\paragraph{Statistical model dependence in estimation from pairwise comparisons} 
Another line of literature on model (in)dependence focuses on estimation from pairwise comparisons. Commonly used statistical models in this setting include the Bradley–Terry–Luce (BTL) ~\citep{bradley1952rank,Plackett75:Analysis} and Thurstone models ~\citep{thurstone1927method}, which all fall under the class of parameter-based models (also known as random utility models). However, recent studies ~\citep{shah2016stochastically,shah2018simple,Heckel19:Active} have demonstrated that estimators derived from more general ``permutation-based'' models offer two notable advantages. First, when data are generated from these broader permutation-based models, these estimators exhibit significantly better performance. Second, even when data originate from parameter-based models such as BTL or Thurstone, the guarantees provided by these general estimators are within logarithmic factors of those achieved by estimators specifically tailored to parameter-based models.


\section{Conclusions and Future Work}\label{sec:conclusion}
In this paper, we introduced the problem of finding good rule picking rules, and argued for its importance.  We also introduced a specific rule, $\aba$, and showed via theoretical analysis and experiments that $\aba$ is able to attain the benefits of both the axiomatic and statistical approaches. As demonstrated by our experiments in \Cref{sec:experiments}, $\aba$ provides a principled way for users to evaluate the impact of modifications to existing methods by running $\aba$ on their own data. In some cases, the answer is that it is an improvement: both changes to F1 rules significantly reduced disagreement (\Cref{fig:f1_ave}). In others, the proposed modification hurts consistency: \citet{Meyer22:Analysis} consider adding outlier rejection to Borda for peer review, but our experiment suggests this method (Trimmed Borda) increases disagreement on their data. By running similar experiments on their own datasets, implementers of our method can determine which aggregation rules, among the ones they are considering, provide the most consistent results for their own applications.

We now discuss several limitations of our approach, and also provide avenues of future work. One of the main limitations is that our approach treats voters' rankings as fixed, whereas in practice voters may choose their rankings based on the RPR being employed. It is of interest to understand RPRs---and $\aba$---more generally under strategic behavior of voters. Such behavior makes it harder to study the impact of a proposed rule change on consistency by running $\aba$ on prior data. 

A separate limitation is that our sampling-based implementation \Cref{sec:experiments} only approximates $\aba$. However, in addition to its desirable behavior on known distributions, our implementation exhibits strong axiomatic properties: an experiment we provide in \Cref{app:axiom_violation} shows that even the two axioms failed by the exact $\aba$ (preserving monotonicity and union consistency) are rarely violated in practice by our implementation, on both simulated and real-world data. Still, it is worth exploring theoretical bounds on how likely our sampling-based implementation is to output the true $\aba$ winner from \Cref{def:aba}. In our experiments, we see that even a small number of splits works well in practice, with additional splits leading to qualitatively identical results. Since the disagreement scores of a given rule for several splits of a given profile are i.i.d., bounding the variance in terms of intuitive properties for the profile (capturing, for example, how close it is to being ``borderline'') can allow using standard tail bounds for bounding the required number of splits. We believe identifying the correct properties of a profile for such a statistical analysis is a worthwhile direction.

Another future avenue is to leverage the generality of \Cref{alg:aba} to study $\aba$ on different output formats (other than rankings) and different distance functions (other than Kendall-Tau distance). For example, in practice, ALMA rankings from \Cref{sec:alma_rank} are used for determining $k$ proposals to be funded for some fixed $k<|\cand|$. Thus, running $\aba$ to pick among multi-winner rules (that return a subset of $k$ alternatives) by using a distance function---such as Jaccard dissimilarity---that compare such subsets may be more appropriate for such data. In \Cref{app:jaccard} we find that experiments using Jaccard dissimilarity provide qualitatively similar results to those in \Cref{sec:alma_rank} when used on the same data. Further, randomized decisions (methods that output a distribution over subsets of $k$ winners) based on the peer-review process are gaining popularity~\citep{heyard2022rethinking, goldberg2025principled}, adding to the richness of this setting. 

Similar questions exist for single-winner elections: while the 0-1 loss (is the same candidate picked on both sides of a split?) emerges as a natural option for measuring disagreement, one can also use other notions of distance between the alternatives chosen on the two sides of a split. For example, \citet[Prop.\ 26]{Berker25:Independence} introduce the ``clone distance'' between any pair of candidates, which measures the dissimilarity of the candidates from the perspective of the voters. Crucially, measuring this distance requires no external information about the candidates, only their positions in the profile. Another alternative would be to repeatedly apply the single-winner rule to the profile, removing the winner after each step. By comparing the resulting rankings from the two sides using another weighted variant of the Kendall-Tau distance (similar to \citet{Kumar10:Generalized}) the implementer can then balance between the two extremes (0-1 loss vs.\ full ranking) by deciding how much weight to put on the first position of the ranking relative to the others. Overall, studying the tradeoff between various distance functions is important for future applications of $\aba$.

Lastly, while we developed and studied several axioms, future work can investigate RPRs under other axioms, either translated from social choice axioms or novel axioms for rule picking rules. Such axioms may inform the design of novel RPRs that focuses on properties other than consistency. Overall, RPRs (and $\aba$) open the door to principled, data-driven rule selection, paving the way for diverse real-world applications.

\section*{Acknowledgements}
We thank John Carpenter and Andrea Corvillon from the Atacama Large Millimeter/submillimeter Array (ALMA), as well as Kate Larson and Ariel Procaccia for helpful discussions. We also thank Sinan Karaböcüoğlu and Boğaçhan Arslan for implementing our method for the \emph{2nd Computational Social Choice Competition} at the 33rd International Joint Conference on Artificial Intelligence (IJCAI 2024). R.E.B.\ and V.C.\ thank the Cooperative AI Foundation, Macroscopic Ventures (formerly Polaris Ventures
/ the Center for Emerging Risk Research), and Jaan Tallinn's donor-advised fund at Founders Pledge for financial support. R.E.B.\ is also supported by the Cooperative AI PhD Fellowship. The work of N.S.\ was supported by NSF CAREER award 1942124 and ONR N000142512346.
\bibliography{dairefs}

\appendix

\section{Omitted Proofs}\label{app:proofs}

In this section, we give the proofs that were omitted in the main body of the paper.

\subsection{Proof of \Cref{thm:axiombundle}}\label{app:consistency}
We first recall our main result from \Cref{sec:consistency}, regarding the consistency axioms for RPRs. 
\axiombundle*

We prove each claim in the theorem as a seperate proposition. 
\begin{restatable}{prop}{reverse}\label{prop:reverse}
    $\aba$ satisfies reversal symmetry. 
\end{restatable}

\begin{proof}
    Given any positional scoring rule $f_s \in F$ and a profile $\profile$, say $f_{s'}$ is the reverse rule of $f_s$, \emph{i.e.}, $s'=(s'_1,s'_2, \ldots, s'_{\ncand})=(1-s_\ncand,1-s_{\ncand-1}, \ldots, 1-s_{1})$. The total scores assigned to any alternative $a \in \cand$ by $f_{s'}$ when run of $\rev(\profile)$ (see \Cref{def:pos}) is 
    \begin{align*}
        t_{\rev(\profile)}^{s'}[a]= \sum_{j=1}^{\ncand} s'_j M_{\rev(\profile)}[a,j]= \sum_{j=1}^{\ncand} (1-s_{\ncand+1-j}) M_{\profile}[a,\ncand+1-j]=\sum_{j=1}^{\ncand} (1-s_{j}) M_{\profile}[a,j]=\nvoters-t^s_{\profile}[a],
    \end{align*}
    implying that $f_s(\profile) =\rev ( f_{s'}(\rev(\profile)))=\rev ( \rev(f_s)(\rev(\profile)))$. Additionally, we note that the Kendall-Tau distance (\Cref{def:kt}) is symmetric with respect to reversals, \emph{i.e.}, $KT(r_1,r_2) = KT (\rev(r_1),\rev(r_2))$. Thus, 
    \begin{align*}
        KT(f_s(\profile^{(1)}),f_s(\profile^{(2)})) 
        &=KT(\rev ( \rev(f_s)(\rev(\profile^{(1)}))),\rev ( \rev(f_s)(\rev(\profile^{(2)}))))\\&= KT( \rev(f_s)(\rev(\profile^{(1)})), \rev(f_s)(\rev(\profile^{(2)}))) .
    \end{align*}
This implies that if $f_s$ minimizes \eqref{eq:aba} in \Cref{def:aba} over random splits of $\profile$, then $\rev(f_s)$ minimizes it over random splits of $\rev(\profile)$, and therefore $\aba(\rules,\profile) = \rev( \aba(\rules, \rev(\profile)))$.
\end{proof}

\begin{restatable}{prop}{welfare}\label{prop:welfare-fails}
    Any welfare-maximizing RPR $\rpr$ fails plurality-shuffling consistency.
\end{restatable}

\begin{proof}
Take any $f_s \in F_S$. We will show that for any $k \in \mathbb{Z}^+$, we have $f_s(\shuffle{[2,\ncand]}{k})=f_p(\profile)$, where $f_p$ is plurality. Recall from \Cref{def:pos} we have $s_1=1$ and $s_\ncand=0$ WLOG. By \Cref{def:pos,def:shuffle}, the total score assigned to a given alternative $a \in \cand$ by $f_s$ on input $\shuffle{[2,\ncand]}{k}$ (say $t_k[a]$) is
\begin{align*}
    t_k[a]&= M_{\profile}[a,1]\cdot k\cdot \ncand !s_1 + \sum_{i=2}^{\ncand} M_{\profile}[a,i]\cdot  \sum_{j=2}^{\ncand}  \frac{k \cdot \ncand!}{(\ncand-1)}s_j \\&=M_{\profile}[a,1]\cdot k\cdot \ncand ! +\frac{k \cdot \ncand!}{(\ncand-1)} \left( \sum_{j=2}^{\ncand-1} s_j\right) (\nvoters-M_{\profile}[a,1])\\
    &= \frac{k \cdot \ncand!}{(\ncand-1)} \left( \sum_{j=2}^{\ncand-1} s_j\right) \nvoters +M_{\profile}[a,1] \cdot  k \cdot \ncand ! \left(1-\frac{\sum_{j=2}^{\ncand-1} s_j}{\ncand-1}\right) \\
    & \equiv C + D \cdot M_{\profile}[a,1],
\end{align*}
where $C$ and $D$ does not depend on $a$. Since $1 =s_1 \geq s_2 \geq \ldots \geq s_\ncand = 0$, we have $\sum_{j=2}^{\ncand-1} s_j < m-1$ and therefore $D>0$. This implies that for any $a,b \in \cand$ we have
\begin{align*}
    a \succ_{f_s(\shuffle{[\ncand-1]}{k})} b &\Leftrightarrow t_k[a] > t_k[b] \\&\Leftrightarrow M_{\profile}[a,1] > M_{\profile}[b,1] \\&\Leftrightarrow a \succ_{f_p(\profile)} b.
\end{align*}
    Since $f_s$ was arbitrarily chosen, this implies all rules in $F_S$ return the same output on $\shuffle{[\ncand-1]}{k}$. Hence, for any welfare-maximizing RPR $Z$ and any $\rules \subseteq \rules_S$, we have $Z(F,\shuffle{[\ncand-1]}{k})=F$, showing that the rule fails \Cref{def:psc}. 
\end{proof}
Since all $f_s \in F_S$ produce the same output on $\shuffle{[\ncand-1]}{k}$, one might wonder whether it matters which rule an RPR returns. However, in practice an RPR might be employed to ``lock in'' a rule for future use, allowing us to use it in future profiles coming from the same source without re-running the RPR at each round, potentially increasing interpretability. Hence, with a plurality-shuffling-consistent rule, we would be able to keep on using $f_p$ in future profiles that are approximately (but not exactly) shuffled, as future repetitions of the same process are likely to once again have their signal concentrated in the top position. Moreover, identifying that no $f_s \in F_S$ provides any information beyond $f_p$ has a computational advantage: while the output of $f_p$ can be computed in $O(n+m)$ time, an arbitrary $f_s \in F_S$ might require $O(mn)$ time.

\begin{restatable}{prop}{pscprop}\label{prop:psc}
    $\aba$ satisfies plurality-shuffling consistency.
\end{restatable}

\begin{proof}
    Fix some finite $\rules \subset \rules_S$. If $|\rules|=1$ we are done. Otherwise, given any profile $\profile$ (that satisfies the condition in \Cref{def:psc}, \emph{i.e.}, $f_p(\profile)$ contains no ties) with $\nvoters=n$ voters an $\ncand=m \geq 3$ alternatives,\footnote{If $m=2$, there is a single (monotonic) positional scoring rule satisfying \Cref{def:pos}: the one with score vector $(1,0)$, which is also equivalent to $f_p$.} define $\profile_k = \shuffle{[2,m]}{k}$ for all $k$. Let $M$ be the matrix such that for each alternative $a \in  \cand$ and index $i \in [m]$, $M[a,i]$ indicates the number of voters in $\profile$ that rank $a$ in the $i^{\text{th}}$ position, and let $M_k$ be the analogous matrix for $\profile_k$.  

    For any $f \in F$ and any two candidates $a, b \in \cand$, say $\dis{f,k}{a,b}$ is the indicator random variable that the two output rankings resulting from applying $f$ to the two sides of a random split of $\profile_k$ into $\profile^{(1)}_k$ and  $\profile^{(2)}_k$ (where each voter is independently and uniformly placed in one of the two sets, as defined in \Cref{def:aba}) disagree about $a$ and $b$'s position; \emph{i.e.},

    $$\dis{f,k}{a,b}= \mathbb{I}[(a \succ_{f(\profile^{(1)}_k)}b\text{ and }b \succ_{f(\profile^{(2)}_k)} a)\text{ or }(b \succ_{f(\profile^{(1)}_k)}a\text{ and }a \succ_{f(\profile^{(2)}_k)} b)].$$
    Similarly, say $\tie{f,k}{a,b}$ is the indicator random variable that $f$ ties $a$ and $b$ on one of the two profiles; \emph{i.e.},
    $$T_{r_1,r_2}^{a,b} = \mathbb{I}[(a \succeq_{f(\profile^{(1)}_k)}b\text{ and }b \succeq_{f(\profile^{(1)}_k)} a)\text{ or }(a \succeq_{f(\profile^{(2)}_k)}b\text{ and }b \succeq_{f(\profile^{(2)}_k)} a)].$$
    Then, using linearity of expectation and the fundamental bridge, we get
    \begin{align}
        \expc \left[KT(f(\profile_k^{(1)}),f(\profile^{(2)}_k))\right] =  \expc \left[\sum_{a,b \in \cand} \dis{f,k}{a,b} + \frac{1}{2} \tie{f,k}{a,b} \right]= \sum_{a,b \in \cand} \Pr[\dis{f,k}{a,b}=1]+ \frac{1}{2}\Pr[\tie{f,k}{a,b}=1].\label{eq:sum_dis}
    \end{align}    
    The rest of the proof will follow from the next lemma.
    \begin{lemma}\label{lemma:shuffle}
        For any $f \in F \setminus \{f_p\}$ there exists a $k_f \in \mathbb{Z}^+$ such that
        \begin{enumerate}
            \item for all $k \geq k_f$ we have $\Pr[\dis{f,k}{a,b}=1]+ \frac{1}{2}\Pr[\tie{f,k}{a,b}=1] > \Pr[\dis{f_p,k}{a,b}=1]+ \frac{1}{2}\Pr[\tie{f_p,k}{a,b}=1]$ for all $a,b \in \cand$ such that $M[a,1]>0$ and $M[b,1]>0$;
            \item If there exists $c \in \cand$ such that $M[c,1]=0$, then $\Pr[\dis{f,k}{a,c}=1]+ \frac{1}{2}\Pr\left[\tie{f,k}{a,c}=1\right] \geq \Pr[\dis{f_p,k}{a,c}=1]+ \frac{1}{2}\Pr[\tie{f_p,k}{a,c}=1]$ for all $a \in \cand$ and $k \in \mathbb{Z}^+$.
        \end{enumerate}
    \end{lemma}

    Since $f_p(\profile)$ contains no ties, there can be at most one $c \in \cand$ with $M[c,1]=0$. Since $m \geq 3$, the lemma implies that  $ \sum_{a,b \in \cand} \Pr[\dis{f,k}{a,b}=1]+ \frac{1}{2}\Pr[\tie{f,k}{a,b}=1]> \sum_{a,b \in \cand} \Pr[\dis{f_p,k}{a,b}=1]+ \frac{1}{2}\Pr[\tie{f_p,k}{a,b}=1]$ for all $k \geq k_f$. By definition of $\aba$, this implies that for all $k \geq k_f$ we have $f \notin \aba(\rules, \shuffle{[m-1]}{k})$, since (\ref{eq:sum_dis}) is strictly greater for $f$ than for $f_p$. Since $|F|$ is finite, picking $k  = \max_{f \in F \setminus\{f_p\}}k_f$ will ensure $\aba(\rules, \shuffle{[m-1]}{k'}=\{f_p\}$ for all $k' \geq k$, as desired. We will now prove the lemma.
    \begin{proof}[Proof of \Cref{lemma:shuffle}]
        For any $f \in F$, say $\{s^f_i\}_{i \in [m]}$ is the scoring vector associated with $f$. Fix any two candidates $a,b \in \cand$. Say $M[a,1]=p_a$ and $M[b,1]=p_b$, with $p_a>p_b$ WLOG.\footnote{We cannot have $p_a=p_b$ as by assumption $f_p(\profile)$ contains no ties.} By \Cref{def:shuffle}, we have $M_k[a, 1] = kp_am !$ and $M_k[b, 1] = kp_bm ! $, whereas $M_k[a, i] = \frac{k(n-p_a)m !}{m-1}$ and $M_k[b, i] = \frac{k(n-p_b)m !}{m-1}$ for all $1< i\leq m$. Similarly, let $M_k[(a,i);(b,j)]$ be the number of voters in $\profile_k$ that rank $a$ in the $i$th position \emph{and} $b$ in the $j$th position (with $i \neq j$). Then
        \begin{align*}
            M_k[(a,i);(b,j)] = \begin{cases}
               k p_a C_1&\text{if }i=1\\
               k p_b C_1&\text{if }j=1\\
               k C_2  &\text{otherwise}
            \end{cases},
        \end{align*}
        where $C_1 \equiv \frac{m!}{m-1}$ and $C_2 \equiv \frac{(n-p_a-p_b)m !}{(m-1)(m-2)} $. Given a random split of $\profile_k$ into $\profile_k^{(1)}$ and $\profile_k^{(2)}$, say $X^{f,k}_i$ is the random variable indicating the total score of $a$ minus the total score of $b$ in $\profile^{(i)}_k$ according to $f$ for each $i \in \onetwo$. Naturally, regardless of the split, $X^{f,k}_1+X^{f,k}_2$ must add up to the total score of $a$ minus the total score of $b$ in $\profile_k$ according to $f$, so
        \begin{align*}
            X^{f,k}_1 +X^{f,k}_2 &= \sum_{i=1}^m s^f_i M_k[a,i] - \sum_{i=1}^m s^f_i M_k[b,i]\\ &= kp_am !-kp_bm! + \left(\frac{k(n-p_a)m !}{m-1}- \frac{k(n-p_b)m !}{m-1}\right) \sum_{i=2}^{m} s^f_i \\&= \frac{\lambda^fk(p_a-p_b)m !}{m-1}
        \end{align*}
        where $\lambda^f= (m-1) - \sum_{i=2}^{m} s^f_i>0$, as $s^f_m=0$. Then, using symmetry of the two sides and the fact that $p_a>p_b$, we get
        \begin{align*}
            \Pr[\dis{f,k}{a,b}=1] = \Pr[X^{f,k}_1\cdot X^{f,k}_2 < 0] &= \Pr\left[X^{f,k}_1\cdot \left(\frac{\lambda^fk(p_a-p_b)m !}{m-1}-X^{f,k}_1\right) < 0\right] \\
            &= \Pr\left[X^{f,k}_1 < 0\right] + \Pr\left[X^{f,k}_1 > \frac{\lambda^f(p_a-p_b)m !}{m-1}\right] \\ &= 2\Pr[X^{f,k}_1<0].
        \end{align*}
        A similar analysis gives $ \Pr[\tie{f,k}{a,b}=1]= 2\Pr[X^{f,k}_1=0]$. Now, if $a$ and $b$ are ranked $i$th and $j$th in a ranking, respectively, the score assigned to $a$ minus that assigned to $b$ by $f$ for that ranking is $s^f_i-s^f_j$. Say $\text{Bin}({z})$ is the fair binomial distribution where $z$ is the number of experiments and the probability of success for each experiment is $1/2$. Then,
        \begin{align}
            &X^{f,k}_1 \sim \sum_{i \neq j \in [m]} (s^f_i-s^f_j)  \fbinom{M_k[(a,i);(b,j)]}\nonumber \\
            &= \sum_{i < j \in [m]} (s^f_i-s^f_j) \left( \fbinom{M_k[(a,i);(b,j)]}-\fbinom{M_k[(a,j);(b,i)]}\right)\nonumber \\
            &= \sum_{i < j \in [m]} (s^f_i-s^f_j) \left( \fbinom{M_k[(a,i);(b,j)]}+\fbinom{M_k[(a,j);(b,i)]}-M_k[(a,j);(b,i)] \right) \nonumber\\
            &= \sum_{i < j \in [m]} (s^f_i-s^f_j) \left( \fbinom{M_k[(a,i);(b,j)]+M_k[(a,j);(b,i)]}-M_k[(a,j);(b,i)] \right) \nonumber\\
            &= \sum_{i < j \in [2,m]} (s^f_i-s^f_j) \left( \fbinom{2kC_2}-kC_2\right) + \sum_{i=2}^{m} (1-s^f_i) \left( \fbinom{k(p_a+p_b)C_1}-kp_bC_1\right) \label{eq:dist_f}        \end{align}
        In particular, for plurality, \eqref{eq:dist_f} gives us:
        \begin{align}
            X^{f_p,k}_1 \sim \sum_{i=2}^{m}  \left( \fbinom{k(p_a+p_b)C_1}-kp_bC_1\right)=\fbinom{k(p_a+p_b)m !}-kp_bm ! \label{eq:dist_veto}
        \end{align}

        Now, fixing some $f \in F \setminus \{f_p\}$, using \eqref{eq:dist_f}, we write $X^{f,k}_1 = Y+ \sum_{i=2}^{m}(1-s^f_i)  Z_i$, where
        $$
        Y \sim  \sum_{i<j \in [2,m]} (s^f_i-s^f_j) \left( \fbinom{2kC_2}-kC_2\right) \text{, and }$$ $$Z_i \sim \left( \fbinom{k(p_a+p_b)C_1}-kp_bC_1\right) \text{ for all }i \in [2,m],
        $$ 
     each sampled independently. It is easy to check that $\expc{[Y]}=0$ and that $Y$ is a symmetric distribution. We will treat two cases separately.

        \noindent\textbf{Case 1:} $p_b=0$. In this case, we cannot have $X_1^{f_p,k} < 0$, as seen by (\ref{eq:dist_veto}). Therefore, 
        \begin{align}
            \Pr[\dis{f_p,k}{a,b}=1]+ \frac{1}{2}\Pr[\tie{f_p,k}{a,b}=1] =0+ \Pr[X^{f_p,k}_1=0]= 2^{-kp_am !}. \label{eq:case1_veto}
        \end{align}
        For $f$, on the other hand, we have
        \begin{align}
            \Pr[X_1^{f,k} = 0] \geq \Pr[Y=0; Z_i=0\text{ for all }i \in [2,m]]= \frac{\Pr[Y=0]}{\left(2^{kp_aC_1}\right)^{m-1}}={\Pr[Y=0]}\cdot {2^{-kp_am !}},\nonumber \\
            \text{ and } \Pr[X_1^{f,k} < 0] \geq \Pr[Y < 0; Z_i=0\text{ for all }i \in [2,m]]=\frac{1-\Pr[Y = 0]}{2} \cdot 2^{-kp_am !},\nonumber \\
            \Rightarrow  \Pr[\dis{f_p,k}{a,b}=1]+ \frac{1}{2}\Pr[\tie{f_p,k}{a,b}=1] =  2\Pr[X_1^{f,k} < 0] + \Pr[X_1^{f,k} = 0]  \geq 2^{-kp_am !} \label{eq:case1_f}.
        \end{align}
        Since $a \in A$ and $k \in \mathbb{Z}^+$ were arbitrarily chosen, combining (\ref{eq:case1_veto}) and (\ref{eq:case1_f}) gives us statement 2 from the lemma. 

        \noindent \textbf{Case 2:} $p_b>0$. Say $\biprob{N}{r}$ is $\Pr[A \leq rN]$ for $A \sim \fbinom{N}$, given that $0< r < \frac{1}{2}$ and $rN$ is an integer. By \citet[Lemma 4.7.2]{Ash90:Information}, we know that
        \begin{align}
            \frac{2^{-N(1-h(r))}}{\sqrt{8Nr(1-r)}} \leq \phi(N,r) \leq 2^{-N(1-h(r))}, \label{eq:binom_tail}
        \end{align}
        where $h(r)=-r\log r - (1-r) \log(1-r)$ is the binary entropy function. We will use \eqref{eq:binom_tail} to upper bound the error probability for $f_p$ and lower bound the error probability for $f$. For convenience, define 
        \begin{align*}
            r_p \eqdef \frac{p_b}{p_a+p_b}, \quad h_p \eqdef h(r_p),  \quad N_p\eqdef(p_a+p_b)C_1.
        \end{align*}
        First we start with $f_p$. By (\ref{eq:dist_veto}) and (\ref{eq:binom_tail}), we have
        \begin{align*}
            \Pr[\dis{f_p,k}{a,b}=1]+ \frac{1}{2}\Pr[\tie{f_p,k}{a,b}=1] &=  2\Pr[X_1^{f_p,k} < 0] + \Pr[X_1^{f_p,k} = 0] \\& \leq  2\Pr[X_1^{f_p,k} \leq 0] = 2\biprob{k(p_a+p_b)m !}{r_p} \\
            & \leq 2^{1-kN_p(m-1)(1-h_p)} \equiv U(k),
        \end{align*}
    where we will use $U(k)$ as an upper bound that depends on $k$. Now consider $f$. Since $f \neq f_p$ and since $1=s^f_1\geq s^f_2\geq \ldots s^f_m =0$, we must have $s^f_2>0$. Fix some $\varepsilon > 0$ such that $\varepsilon < \min\left(\frac{p_a-p_b}{2(p_a+p_b)}, \frac{p_b}{(p_a+p_b)(1-s^f_2)}\right)$ (if $s^f_2=1$, then ensuring $\varepsilon< \frac{p_a-p_b}{2(p_a+p_b)}$ is sufficient). Recalling that $X^{f,k}_1 = Y+ \sum_{i=2}^{m}(1-s^f_i) Z_i$ and that $s^f_m=0$, we have
        \begin{align*}
            \Pr[X_1^{f,k} \leq 0] &\geq \Pr[Y \leq 0] \cdot  \Pr[Z_{2} \leq  \varepsilon k N_p  ]   \cdot  \Pr[Z_{m} \leq - \varepsilon k N_p (1-s^f_2)  ]  \cdot \prod_{i =3}^{m-1} \Pr[Z_i \leq 0],\\
            \Pr[X_1^{f,k} < 0] &\geq  \Pr[Y < 0] \cdot  \Pr[ Z_{2} \leq  \varepsilon k  N_p  ]   \cdot  \Pr[Z_{m} \leq - \varepsilon k N_p (1-s^f_2)  ]  \cdot \prod_{i =3}^{m-1} \Pr[Z_i \leq 0]\\&=  \Pr[Y > 0] \cdot   \Pr[Z_{2} \leq  \varepsilon k N_p ]   \cdot  \Pr[Z_{m} \leq - \varepsilon k N_p (1-s^f_2)  ]  \cdot \prod_{i =3}^{m-1} \Pr[Z_i \leq 0];
        \end{align*}
        therefore, 
        \begin{align*}
             \Pr[\dis{f,k}{a,b}=1]+ \frac{1}{2}\Pr[\tie{f,k}{a,b}=1] &=  2\Pr[X_1^{f,k} < 0] + \Pr[X_1^{f,k} = 0] \\&= \Pr[X_1^{f,k} < 0] + \Pr[X_1^{f,k} \leq 0] \\
             &\geq\Pr[Z_{2} \leq  \varepsilon k N_p ]   \cdot  \Pr[Z_{m} \leq - \varepsilon k N_p (1-s^f_2)  ]  \cdot \prod_{i =3}^{m-1} \Pr[Z_i \leq 0] \\
             &\geq \Pr[Z_{2} \leq  \lfloor \varepsilon k N_p \rfloor ]   \cdot  \Pr[Z_{m} \leq - \lceil \varepsilon k N_p (1-s^f_2) \rceil  ]  \cdot \prod_{i =3}^{m-1} \Pr[Z_i \leq 0] \\
            & \geq \biprob{kN_p}{r_2(k)} \cdot \biprob{kN_p}{r_m(k)} \cdot (\biprob{kN_p}{r_p} )^{m-3}  
        \end{align*}
        where  $r_2(k) = r_p +\frac{\lfloor \varepsilon k N_p  \rfloor}{kN_p} $ and  $r_m(k) =r_p -\frac{\lceil \varepsilon k N_p (1-s^f_2) \rceil}{kN_p}$. As $k\rightarrow \infty$, we have $r_2(k) \rightarrow r_p+\varepsilon \in \left(0,\frac{1}{2}\right)$ and   $r_m(k) \rightarrow r_p-\varepsilon(1-s^f_2) \in \left(0,\frac{1}{2}\right)$; hence, we can choose $k$ large enough such that $0 < r_i < 1/2$ for both $i \in \{2,m\}$. Further, by construction, $kN_p \cdot r_i(k)$ is an integer for all $k \geq 0$ and $i \in \{2,m\}$. Define $h_i(k) \eqdef h_i(r_i(k))$ for $i \in \{2,m\}$. By \eqref{eq:binom_tail}, we then have
        \begin{align*}
            \Pr[\dis{f,k}{a,b}=1]+ \frac{1}{2}&\Pr[\tie{f,k}{a,b}=1] \geq \biprob{kN_p}{r_2(k)} \cdot \biprob{kN_p}{r_m(k)} \cdot (\biprob{kN_p}{r_p} )^{m-3}   \\ &\geq \frac{2^{-kN_p(1-h_2(k))}}{\sqrt{8kN_pr_2(k)(1-r_2(k))}} \cdot \frac{2^{-kN_p(1-h_m(k))}}{\sqrt{8kN_pr_m(k)(1-r_m(k))}}\cdot \left(   \frac{2^{-kN_p(1-h_p)}}{\sqrt{8kN_pr_p(1-r_p)}} \right)^{m-3}\\
            & =    \frac{2^{-kN_p(m-1-h_2(k)-h_m(k)-(m-3)h_p)}}{\sqrt{(8kN_p)^{m-1} r_2(k)(1-r_2(k))r_m(k)(1-r_m(k))(r_p(1-r_p))^{m-3}}} \equiv L(k),  
        \end{align*}
       where we will use $L(k)$ as a lower bound that depends on $k$. Now, we would like to show that there exists $k_f^{a,b}$ such that $L(k)>U(k)$ for all $k \geq k_f^{a,b}$. We have
        \begin{align} \label{eq:ratio}
            \frac{L(k)}{U(k)}=\frac{2^{kN_pP(k)-1}}{\sqrt{(8kN_p)^{m-1} r_2(k)(1-r_2(k))r_m(k)(1-r_m(k))(r_p(1-r_p))^{m-3}}},
        \end{align}
        where $P(k)=(m-1)(1-h_p)-(m-1-h_2(k)-h_m(k)-(m-3)h_p)=h_2(k)+h_m(k) - 2 h_p$. We would like to show $\lim_{k \rightarrow \infty} \frac{L(k)}{U(k)} = \infty$, giving us our desired relationship. We recall that $r_2(k) \rightarrow r_p+\varepsilon$ and $r_2(k) \rightarrow r_p-\varepsilon(1-s^f_2)$ as $k\rightarrow \infty$, so the denominator of \eqref{eq:ratio} scales as $\Theta(k^{\frac{m-1}{2}})$ for large $k$. As for the nominator, we have $ \lim_{k \rightarrow \infty} P(k) = h(r_p+\varepsilon)+h(r_p-\varepsilon(1-s^f_2))-2h(r_p) \equiv H(\varepsilon)$. We have $H(0)=0$ and $H'(0)=h'(r_p)-(1-s^f_2)h'(r_p)=s^f_2h'(r_p)>0$ since $s^f_2>0$ and
        \begin{align*}
            h'(r_p) =  -\log r_p -1 +\log(1-r_p) +1 = \log\left(\frac{1}{r_p} -1\right) > \log(2-1) = 0,
        \end{align*}
        as $r_p = \frac{p_b}{p_a+p_b} < \frac{1}{2}$. We can therefore set $\varepsilon >0$ small enough such that $H(\varepsilon)>0$. Then the nominator of \eqref{eq:ratio} scales as $\Theta(2^{kN_pH(\varepsilon)})$ for large $k$, dominating the denominator. Hence, $\lim_{k \rightarrow \infty} \frac{L(k)}{U(k)} = \infty$, as desired. Therefore, we can choose a large enough $k_f^{a,b}$ such that for all $k \geq k_{f}^{a,b}$ we have
        \begin{align*}
            \Pr[\dis{f,k}{a,b}=1]+ \frac{1}{2}\Pr[\tie{f,k}{a,b}=1] \geq L(k) > U(k) \geq \Pr[\dis{f_p,k}{a,b}=1]+ \frac{1}{2}\Pr[\tie{f_p,k}{a,b}=1].
        \end{align*}
        Since there are finitely many candidates $\cand$, picking $k_f = \max_{(a,b) \in \cand^2} k_f^{a,b}$ is sufficient for proving statement 1 from the lemma.
    \end{proof}
\end{proof}

 \begin{restatable}{prop}{union}\label{prop:impossibility-union}
     No (anonymous) RPR can satisfy all three of reversal symmetry, plurality-shuffling consistency (PSC), and union consistency (UC). 
 \end{restatable}
 
 \begin{proof}
    Assume an RPR $\rpr$ satisfies PSC and reversal symmetry. We will show that $\rpr$ necessarily fails UC. Fix $\cand=\{a,b,c\}$ and $F=\{f_p,f_v\}$. Say $\profile_a$ consists of two voters ranking $a \succ b \succ c$ and one voter ranking $b \succ a \succ c$. Similarly, say $\profile_b$ consists of two voters ranking $c \succ b \succ a$ and one voter ranking $b \succ c \succ a$. We have $f_p(\profile_a) =  a \succ b \succ c$ and $f_p(\profile_b)= c \succ b \succ a$, so both profiles satisfy the preconditions of \Cref{def:psc}. By PSC, there exist $k_a,k_b$ such that for all $k'_a \geq k_a$ and $k'_b \geq k_b$ we have $\rpr(\rules, \shufflep{[2,m]}{k'_a}{\profile_a})=\rpr(\rules, \shufflep{[2,m]}{k'_b}{\profile_b})=\{f_p\}$. Pick $k =\max\{k_a,k_b\}$ and let $\profile = \shufflep{[2,m]}{k}{\profile_a} +  \shufflep{[2,m]}{k}{\profile_b}$. It is straightforward to check that $\profile=\rev(\profile)$ up to a permutation of voters, with each $r \in \tL(\cand)$ appearing $6k$ times. By anonymity and reversal symmetry, this implies $\rpr(\rules, \profile) = \rev(\rpr(\rules, \profile))$, which is only possible if $\rpr(\rules, \profile)= \{f_v,f_p\}$, proving UC is violated. 
\end{proof}

\subsection{Proof of \Cref{thm:preservebundle}}\label{app:rpr_preserve}
\preservebundle*
We prove each claim as a separate proposition. 

\begin{restatable}{prop}{anoneut}\label{prop:aba_anoneut}
    $\aba$ preserves anonymity \& neutrality.
\end{restatable}
\begin{proof}
    The proof follows definitionally from the anonymity and neutrality of the Kendall-Tau distance (\Cref{def:kt}), and therefore of the $\aba$ function (\Cref{def:aba}).
\end{proof}

Next, we formally define the social choice axioms given in part (2) of \Cref{thm:preservebundle}. Given a profile $\profile$ and $a,b \in \cand$, we say $a$ pairwise defeats $b$ if $| \{i \in N: a \succ_{\sigma_i} b \}| > | \{i \in N: b \succ_{\sigma_i} a \}|$. Then, we say an SWF $f \in F$ satisfies\ldots
\begin{itemize}
    \item \ldots the \emph{Smith Criterion (SC)} if the alternative(s) ranked highest in $f(\profile)$ belongs to the Smith set of $\profile$, \emph{i.e.}, the smallest set $S \subseteq \cand$ such that every $a \in S$ pairwise defeats every $b \in \cand \setminus S$.
    \item \ldots\emph{Condorcet Consistency (CC)} if it satisfies the SC for all profiles $\profile$ that have a Smith set containing a single alternative (called the Condorcet winner).
    \item \ldots\emph{Majority Winner (MW)} if is satisfies CC whenever the Condorcet winner is the top-ranked candidate of a majority ($>\ncand$) of voters.
    \item \ldots\emph{Pairwise Majority Consistency (PMC)} if whenever there exists $r \in \wL(\cand)$ such that $a \succ_r b$ if and only if $a$ pairwise defeats $b$ in $\profile$, then $f(\profile)=r$. \citep{Ge24:Axioms}
    \item \ldots \emph{unanimity} if whenever $\sigma_i=r$ for all $i \in \voters$, then $f(\profile)=r$.
\end{itemize}

\begin{restatable}{prop}{rpreserve}\label{prop:rpr_preserve}
    Any RPR $Z$ preserves the Smith criterion, Condorcet consistency, majority winner, pairwise majority consistency, and unanimity. 
\end{restatable}

\begin{proof}
    Like \Cref{prop:rpr_preserve}, the proof follows from the definitions of the axioms. Each of them are of the form ``if $\profile$ satisfies conditions $X$, then $f(\profile)$ must satisfy conditions $Y$.'' Since for each $\profile$ we have $f_\rpr^\rules(\profile)=f(\profile)$ for some $f \in \rules$, restricting $F$ to rules that satisfy this axiom will ensure $f_\rpr^\rules(\profile)$ will satisfy $Y$ whenever $\profile$ satisfies $X$.
\end{proof}

Before proving the general impossibility result in part (3) of \Cref{thm:preservebundle}, we will show that $\aba$ does not preserve monotonicity. 

\begin{ex} \label{ex:aba_nomono}
    Say $\rules =\{f_p, f_v\}$ (plurality and veto). Both of these rules are monotonic, as all positional scoring rules are. Fix $k \in \mathbb{Z}^+$ and consider the following profile $\profile$:
\begin{itemize}
    \item Group 1: $k$ voters rank $a \succ b \succ c\succ d$.
    \item Group 2: $k$ people voted $d\succ c\succ a\succ b$.
    \item Group 3: $k$ people voted $d\succ c\succ a\succ b$.
    \item Group 4: $k$ voters rank $b \succ a \succ c\succ d$.
    \item Group 5: $3k$ voters rank $c \succ a \succ b \succ d$.
    \item Group 6: $3k$ voters rank $d \succ b \succ a\succ c$.
\end{itemize}

The veto scores (the number of voters that rank them bottom) of $a,b,c,d$ are $0,2k,3k,5k$, respectively. By an analogous argument to that in \Cref{ex:aba_run}, $\Pr[f_v(\profile^{(1)})=f_v(\profile^{(2)})= (a \succ b \succ c\succ d)] \rightarrow 1$ as $k \rightarrow \infty$, where the probabilities are taken over the random process in \Cref{def:aba}. Hence, the expectation in Equation (\ref{eq:aba}) of \Cref{def:aba} approaches 0  for $f_v$ as $k$ grows. However, the plurality scores (the number of voters that rank them top) of $a$ and $b$ are tied, implying the the probability that $f_p(\profile^{(1)})$ and $f_p(\profile^{(2)})$ will disagree about $a$ and $b$ converges to 1 as $k \rightarrow \infty$. This implies the expectation in (\ref{eq:aba}) approaches at least 1 for $f_p$ as $k$ grows. Therefore, there exists a $k_1$ such that $\aba(\rules, \profile)=\{f_v\}$ for all $k \geq k_1$, and therefore $f_{\aba}^\rules(\profile)= (a \succ b \succ c\succ d)$.
Now say the voters in group 1 and 2 promote $b$ by a single spot, to get the profile $\profile'$:
\begin{itemize}
    \item Group 1: $k$ voters rank $b \succ a \succ c\succ d$.
    \item Group 2: $k$ people voted $d\succ c\succ b\succ a$.
    \item Group 3: $k$ people voted $d\succ c\succ a\succ b$.
    \item Group 4: $k$ voters rank $b \succ a \succ c\succ d$.
    \item Group 5: $3k$ voters rank $c \succ a \succ b \succ d$.
    \item Group 6: $3k$ voters rank $d \succ b \succ a\succ c$.
\end{itemize}
Importantly, we have $\profile'= \rev(\profile)$ up to permuting Groups 1 with 2, 3 with 4, and 5 with 6. By anonymity and reversal symmetry, this implies $\aba(\rules, \profile')=\{f_p\}$, so  $f_{\aba}^\rules(\profile')= (d \succ c \succ b\succ a)$. Thus, by promoting $b$, we have gotten $\rank_{f_{\aba}^\rules(\profile)}(b) = 2 < 3 = \rank_{f_{\aba}^\rules(\profile')}(b)$, which violates monotonicity. 
\end{ex}

We now show how the same example can be used for other (anonymous) RPRs satisfying reversal symmetry, in order to prove part (3) of \Cref{thm:preservebundle}.
\begin{restatable}{prop}{impmono}\label{prop:impossible-mono}
     No (anonymous) RPR can satisfy reversal symmetry and preserve monotonicity.
\end{restatable}

\begin{proof}
    Say RPR $\rpr$ satisfies reversal symmetry and is anonymous. Take $\profile, \profile'$ and $\rules$ from \Cref{ex:aba_nomono}. We consider two cases.\\
    \noindent\textbf{Case 1:}  $Z(\rules, \profile) = \{f_v\}$. Then, by the same reasoning as in the proof of \Cref{ex:aba_nomono}, we must have $Z(\rules, \profile') = \{f_p\}$ (as that example only uses anonymity and reversal symmetry to show this), and monotonicity is violated by $f^\rules_Z$. \\
    \noindent\textbf{Case 2:}  $Z(\rules, \profile) = \{f_p\}$. In that case, $f_Z^\rules(\profile)=(d \succ c \succ a=b)$, \emph{i.e.}, $a$ and $b$ are tied in the bottom. By reversal symmetry, we have $Z(\rules, \profile') = \{f_v\}$ and therefore $f_Z^\rules(\profile')=(a=b \succ c \succ d)$. However, one can go from $\profile'$ to $\profile$ by simply promoting $a$ by a spot in the rankings of groups 1 and 2. This implies that  $\rank_{f_{\rpr}^\rules(\profile')}(a) = 1 < 3 = \rank_{f_{\rpr}^\rules(\profile)}(a)$, so promoting $a$ results in increasing its rank, therefore $f_{\rpr}^\rules$ is not monotonic. 

    In either case, we see that $f_{\rpr}^\rules$ violates monotonicity. Since $f_p$ and $f_v$ are both monotonic, this proves that $\rpr$ does not preserve monotonicity.
\end{proof}

\subsection{Proof of \Cref{thm:perfpos}} 
\label{app:perfpos}

We recall the complexity result from \Cref{sec:computational}.

\nphard*

We first introduce a generalization of the computational problem $\perfpos$, which we will call $k\-\perfpos$: For each voter $i \in \voters$, we are given a strict ranking $\sigma_i \in \tL(A_i)$ over a \emph{subset} of alternatives $A_i \subseteq A$, with $|A_i|=k$. We are also given a split of voters $\voters=\voters_1 \sqcup \voters_2$ with $|N_1|=|N_2|$. For each $j \in \onetwo, a \in \cand$, and $i \in [\nballot ]$, $\cmatrix_j[a,i]$ indicates the number of voters in $\voters_j$ that rank $a$ in their $i^{\text{th}}$ position. Then, $\perfpos$ asks: is there a positional scoring rule $f_s \in F_S$ that achieves zero disagreement over this split, $\emph{i.e.}$, is there a vector ${s}=(s_i)_{i \in [k]}$ with $1=s_1\geq s_2\geq \ldots \geq s_\nballot  =0$ such that for all $a,b\in \cand$
\begin{align*}
    (T_1[a]-T_1[b])(T_2[a]-T_2[b])>0,
\end{align*}
where $T_j[a]= \sum_{i=1}^\nballot  \cmatrix_j[a,i] s_i$ for any $a \in \cand$ and $j \in \onetwo$. 

Clearly $k\-\perfpos$ contains $\perfpos$ (for $k=\ncand$). However, as we show next, it is not harder: Given an instance of $k\-\perfpos$ (with input profile $\profile$), complete the ranking of each voter $i \in \voters$ to a strict ranking over all alternatives by placing the remaining $\ncand-k$ alternatives in $\cand \setminus \cand_i$ at the bottom of $\sigma_i$, giving rise to a new (complete) profile $\profile'$. Then define $\profile''=\shufflep{[k,\ncand]}{1}{\profile'}$, \emph{i.e.}, the $1-$shuffling of $\profile'$ with respect to positions $[k,\ncand]=\{k,k+1,\ldots,\ncand\}$ (\Cref{def:shuffle}). Then, we claim the answer to the original $k\-\perfpos$ instance with profile $\profile$ is a yes if and only if the answer to the $\perfpos$ instance using $\profile''$ (and the same split as $\profile$) is a yes. Indeed, if there exists a positional scoring rule $f_{s}$ with $1= s_1 \geq s_2 \geq \ldots  \geq s_k=0$ that achieves zero disagreement with $\profile$, then $f_{s''}$ with $s''=(s_i)_{i \in [\ncand]}$ defined as $s''_i=\begin{cases}
    s_i &\text{if }i \leq k\\
    0 & \text{otherwise}
\end{cases}$
achieves zero disagreement with $\profile''$. On the other hand, given a positional scoring rule $f_{s''}$ with $1= s''_1 \geq s''_2 \geq \ldots \geq s''_\ncand=0$ that achieves zero disagreement with $\profile''$, define $x= \frac{s''_k +s''_{k+1}+\ldots + s''_\ncand}{\ncand-k+1}$ and define vector $s=(s_i)_{i \in [k]}$ as $s_i = (s''_i-x)/(1-x)$ for all $i \in [k-1]$ and $s_k=0$ (We have $x<1$ since $s''_\ncand=0$). We will show $f_s$ achieves zero disagreement with $\profile$. Given any $a \in \cand$ and $j \in \onetwo$, the total score assigned by $f_{s''}$ to $a$ on input $\profile''^{(j)}$ (restriction of $\profile''$ to voters in $N_j$) is 
\begin{align*}
    T''_j[a]= \sum_{i=1}^{k-1} k! M_j[a,i] s''_i +  \sum_{i=k}^{\ncand} k! \frac{|N_j|-\sum_{i'=1}^{k-1}M_j[a,i']}{\ncand-k+1} s''_i =  k!\left(x|N_j|+\sum_{i=1}^{k-1}  M_j[a,i] (s''_i-x)  \right)
\end{align*}
by \Cref{def:shuffle}. The score assigned by $f_{s}$ to $a$ on input $\profile^{(j)}$ (restriction of $\profile$ to voters in $N_j$), on the other hand, is 
\begin{align*}
    T_j[a]= \sum_{i=1}^{k}  M_j[a,i] s_i =\sum_{i=1}^{k-1}  M_j[a,i] \frac{s''_i-x}{1-x}= \frac{T''_j[a]-k!x|N_j|}{k!(1-x)}.
\end{align*}

Since the total score of every alternative is just shifted by a constant and then rescaled by another constant, for any $a,b \in \cand$ we have $T_j''[a]>T_j''[b]$ if and only if $T_j[a]>T_j[b]$. As $f_{s''}$ achieves zero disagreement with $\profile''$, this proves $f_{s}$ achieves zero disagreement with $\profile$. Thus, in the proof of \Cref{prop:kperf} below, we reduce 3SAT to $k\-\perfpos$, which proves the NP-hardness for $\perfpos$ too (membership follows from the fact that positional scoring rules are easy to compute). 

\begin{prop}\label{prop:kperf}
    $k\-\perfpos$ is $\NP$-hard.
\end{prop}

\begin{proof}
    We will be reducing from 3SAT. Say $\phi$ is a 3CNF formula with clauses $C_1,C_2,\ldots,C_\nclause$ and binary variables $x_1,x_2,\ldots, x_\nvars$. For each $\set \in \onetwo$, we will construct and instance of $\perfpos$ by first specifying $\cmatrix_\set[a,i]$ for alternative $a \in \cand$ and position $i \in [k]$, and then explicitly designing rankings that is consistent with that $\cmatrix_\set$. We start with setting $\nballot =\nvars+2$ and $\varepsilon=\frac{1}{7(\nballot+2)}$.

    For each $i \in [\nvars]$, add two candidates $a_i$ and $b_i$, with
    \begin{align*}
        \cmatrix_1[a_i,j]&= \begin{cases}
            1+(\nballot +3)(i-1) &\text{if }j=1\\
            \nballot  +2 &\text{if }j=i+1\\
            0 &\text{otherwise}
        \end{cases} \\\cmatrix_1[b_i,j]&= \begin{cases}
            (\nballot +3)(i-1) &\text{if }j=1\text{ and }i \neq 1 \\
            \nballot +2 &\text{if }j=i\\
            1 &\text{if }j=\nballot \\
            0 &\text{otherwise}
        \end{cases}
    \end{align*}
    and
    \begin{align*}
        \cmatrix_2[a_i,j]&= \begin{cases}
            1+(1/\varepsilon+1)(i-1)  &\text{if }j=1\\
            1/\varepsilon &\text{if }j=i+1\\
            0 &\text{otherwise}
        \end{cases} \\ \cmatrix_2[b_i,j]&= \begin{cases}
            (1/\varepsilon+1)(i-1)  &\text{if }j=1\text{ and }i \neq 1\\
            1/\varepsilon &\text{if }j=i \\
            1 &\text{if }j=\nballot \\
            0 &\text{otherwise}
        \end{cases}.
    \end{align*}
    For each $i \in [\nclause]$, say the clause $C_i$ in $\phi$ consist of variables $\{x_{j}\}_{j \in V_i}$ for $V_i \subseteq  [\nvars]$, with $1 \leq |V_i|\leq 3$. For this clause $C_i$, add two candidates $c_j,d_j$, with
        \begin{align*}
        \cmatrix_1[c_i,j]= \begin{cases}
            \nvars(\nballot +3)+2i &\text{if }j=1\\
            0 &\text{otherwise}
        \end{cases} \quad\quad \quad \cmatrix_1[d_i,j]= \begin{cases}
            \nvars(\nballot +3)+2i-1 &\text{if }j=1\\
            1 &\text{if }j=\nballot \\
            0 &\text{otherwise}
        \end{cases}.
        \end{align*}

For constructing $\cmatrix_2[c_i,j]$ and $\cmatrix_2[d_i,j]$, say $z\in \{0,1,2,3\}$ is the number of negated literals in $C_i$. To build $\cmatrix_2[c_i,j]$ and $\cmatrix_2[d_i,j]$ for each $j \in [\nballot]$, start from the following:
        \begin{align*}
        \cmatrix_2'[c_i,j]& = \begin{cases}
            2z +  \nvars(1/\varepsilon+1) +6(\nballot+3) (i-1)&\text{if }j=1\\
            1 &\text{if }j=\nballot\\
            0 &\text{otherwise}
        \end{cases}\\ \cmatrix_2'[d_i,j]&= \begin{cases}1+
            \nvars(1/\varepsilon+1) +6(\nballot+3)(i-1) &\text{if }j =1\\
            2z &\text{if }j =\nballot \\
            0 &\text{otherwise}
        \end{cases}.
        \end{align*}
        Now for each $j\in V_i$, do the following:
        \begin{itemize}
            \item If $x_{j}$ appears non-negated in $C_i$, then add $2(\nballot+2)$ to $\cmatrix_2'[c_i,j]$ and $2(\nballot+2)$ to $\cmatrix_2'[d_i,j+1]$. 
            \item If $x_{j}$ appears negated in $C_i$, then add $2(\nballot+2)$ to $\cmatrix_2'[c_i,j+1]$ and $2(\nballot+2)$ to $\cmatrix_2'[d_i,j]$.
        \end{itemize}
        Finally, set $\cmatrix_2[c_i,j]$ and $\cmatrix_2[d_i,j]$ to the resulting $\cmatrix_2'[c_i,j]$ and $\cmatrix_2'[d_i,j]$ for each $j\in [\nballot ]$. 

Say $A=\{a_i\}_{i \in [\nvars]}, B=\{b_i\}_{i \in [\nvars]}, C=\{c_i\}_{i \in [\nclause]},$ and $D=\{d_i\}_{i \in [\nclause]}$. We now construct a profile that corresponds to the above $\cmatrix_\set$. First, add $\nballot$ more candidates $E=\{e_i\}_{i \in [\nballot]}$. Now for each $\set \in \onetwo$, each $f \in A \sqcup B \sqcup C \sqcup D$ and each $i \in [\nballot]$, add $\cmatrix_\set[f,i]$ voters to the set $N_\set$ that rank $f$ in the $i^{\text{th}}$ position, and ranks $e_{j}$ in the $j^{\text{th}}$ position for all $j \in [\nballot] \setminus \{i\}$. Say we have added $n_1$ and $n_2$ voters to $N_1$ and $N_2$ so far respectively and that $\set'= \argmax_{\set \in \onetwo}n_\set$ (If it's a tie, pick $\set'$ it arbitrarily). Add $(14\nballot+28) n_{\set'}$ and $(14\nballot+29)n_{\set'}-n_{3-\set'}$ voters to $N_{\set'}$ and $N_{3-\set'}$, respectively, all of whom rank $e_i$ in the $i^{\text{th}}$ position for all $i \in [\nballot]$. This also ensures that $|N_1|=|N_2|$. For each $\set \in \onetwo$ say $\profile^{(\set)}$ is the final vector of rankings of $N_\set$, as specified.

By construction, for each $\set \in \onetwo$, $f \in A \sqcup B \sqcup C \sqcup D $, and $i \in [\nballot]$, $\profile^{(\set)}$ indeed has $\cmatrix_\set[f,i]$ voters that rank $f$ in their $i^{\text{th}}$ position. Further, we have a total of $| A \sqcup B \sqcup C \sqcup D \sqcup E|=2(\nvars+\nclause)+\nballot= 3\nvars+2\nclause+2$ candidates (which polynomial in $\nvars,\nclause$), and $|N_1|=|N_2|=(14\nballot+29)\max_{\set \in \onetwo} \sum_{f \in A \cup B \cup C \cup D} \sum_{i=1}^\nballot \cmatrix_\set[f,i] \leq (14\nballot+29)2(\nvars+\nclause)\nballot \cdot \max_{{f \in A \cup B \cup C \cup D}, i \in [\nballot]}\cmatrix_\set[f,i]$ (which is also polynomial in $\nvars,\nclause$ since all entries of $\cmatrix_\set[f,i]$ are). 

We now claim that $\phi$ is satisfiable if and only if there exists a positional scoring rule that gives full agreement between $\profile^{(1)}$ and $\profile^{(2)}$. 

\noindent $(\Leftarrow):$ Assume there is a positional scoring rule $f_s$ with $s=(s_i)_{i \in \nballot}$ that gives full agreement for $\cmatrix_1$ and $\cmatrix_2$. Say $T_\set[a]= \sum_{i=1}^\nballot  \cmatrix_\set[a,i] s_i$ for all $a \in \cand$ and $\set \in \onetwo$. In particular, we must have agreement between $a_i$ and $b_i$ for each $i\in [\nvars]$. We have
\begin{align*}
    T_1[a_i]=(1+(\nballot +3)(i-1))s_1 +(\nballot+2) s_{i+1}, &\quad 
    T_1[b_i]=((\nballot +3)(i-1))s_1 +(\nballot+2) s_{i}+s_\nballot,\\
    T_2[a_i]=(1+(1/\varepsilon +1)(i-1))s_1 + s_{i+1}/\varepsilon, &\quad 
    T_2[b_i]=((1/\varepsilon +1)(i-1))s_1 +s_i/\varepsilon +s_\nballot.
\end{align*}
Since $s_1=1$ and $s_\ncand=0$, perfect agreement implies that we must have
\begin{align*}
    (T_1[a_i]-T_1[b_i])(T_2[a_i]-T_2[b_i])= (1-(\nballot+2)(s_i-s_{i+1}))(1-(s_i-s_{i+1})/\varepsilon)>0.
\end{align*}
This implies we must either have $s_i-s_{i+1}< \varepsilon$ or $s_i-s_{i+1}> \frac{1}{\nballot+2}$. Set the binary variable $x_i$ to False if $s_i-s_{i+1}< \varepsilon$ and to True if $s_i-s_{i+1}> \frac{1}{\nballot+2}$. We now argue that the resulting $\{x_i\}_{i \in [\nvars]}$ satisfies $\phi$, \emph{i.e.}, satisfies all of its clauses. Fix any $i \in [\nclause]$. We will show that $C_i$ is satisfied. By assumption of full agreement, 
\begin{align*}
    (T_1[c_i]-T_1[d_i])(T_2[c_i]-T_2[d_i])>0.
\end{align*}
Since $T_1[c_i]-T_1[d_i]= (t(\nballot+3)+2i)s_1 - (t(\nballot+3)+2i-1)s_1 -s_\nballot = 1>0 $, this implies that $(T_2[c_i]-T_2[d_i])>0$. Say $\{x_{j}\}_{j \in V_i}$ are the variables that appear in $C_i$ (for $V_i\subseteq [\nvars]$), and that $Z= \{j\in V_i: x_{j}\text{ is negated in }C_i\}$ with $|Z|=z$. This implies we have
        \begin{align*}
            &T_2[c_i]=  ( 2z+\nvars(1/\varepsilon+1) +6(\nballot+3)(i-1))s_1 + s_\nballot+2(\nballot+2)\left(\sum_{g \in V_i\setminus Z }s_{g} +\sum_{g \in Z }s_{g+1} \right),\\ 
            &T_2[d_i]=  (1+ \nvars(1/\varepsilon+1) +6(\nballot+3)(i-1))s_1 + 2zs_\nballot + 2(\nballot+2)\left(\sum_{g \in V_i \setminus Z }s_{g+1} +\sum_{g \in Z }s_{g} \right), \\
            &T_2[c_i]-T_2[d_i]=2z-1+2(\nballot+2)\left(\sum_{g \in V_i \setminus Z}(s_{g}-s_{g+1})-\sum_{g \in Z  }(s_{g}-s_{g+1}) \right).
        \end{align*}
        The only way for $C_i$ to be not satisfied is if $x_{g}$ is assigned to True (\emph{i.e.}, $s_{g}-s_{g+1}>\frac{1}{\nballot+2}$) for all $g\in Z$ and $x_{h}$ is assigned to False (\emph{i.e.},s $s_{h}-s_{h+1}<\varepsilon$) for all $ h\in V_i\setminus Z$. This would imply, however $ T_2[c_i]-T_2[d_i] < 2z-1+2(\nballot+2)\left(\sum_{g \in V_i \setminus Z} \varepsilon -\sum_{g \in Z  }\frac{1}{\nballot+2} \right)= 2z-1+ \frac{2(|V_i|-z)}{7}-2z<-1+\frac{6}{7}<0$, which gives a contradiction.

        As assuming $C_i$ is not satisfied gives a contradiction to the assumption that $s$ gives agreement between $\profile^{(1)}$ and $\profile^{(2)}$ for $c_i$ and $d_i$, $C_i$ must be satisfied. Since this is true for all $i \in [\nclause]$, this implies that $\phi$ is satisfiable.

        \noindent $(\Rightarrow):$ Assume $\phi$ is satisfiable for truth assignments $\{x^*_i\}_{i \in [\nvars]}$. For any $i \in [\nvars]$, define $\delta_i =
        \begin{cases}
            \frac{1}{\nballot+1} &\text{if }x^*_i\text{ is True}\\
            \frac{\varepsilon}{2} &\text{otherwise.}
        \end{cases}$.   
        Define the positional scoring rule $s=(s_i)_{i \in [\nballot]}$ as $s_i= \begin{cases}
            1&\text{if }i=1\\
            1-\left(\sum_{j=1}^{i-1}\delta_j\right)&\text{if }\nballot>i>1\\
            0&\text{if }i=\nballot
        \end{cases}$. Since $s_{\nballot-1}= 1-\left(\sum_{j=1}^{\nballot-2}\delta_j\right)=  \geq  1-\frac{\nballot-2}{\nballot+1}>0=s_\nballot$, monotonicity is satisfied, and $s$ is a valid scoring rule. We will show that $s$ gives perfect agreement between $\profile^{(1)}$ and $\profile^{(2)}$. As always, say $T_\set[f]=\sum_{i=1}^\nballot \cmatrix_\set[f,j]s_i$ for all $\set \in \onetwo$ and $f \in \cand$. Since  $1=s_1\geq s_2 \geq \ldots \geq s_\nballot =0$ and since $s_i>0$ for each $i\in [\nballot-1]$, the total scores for each $i\in [\nvars]$ and $j\in[\nclause]$ are 
        \begin{align}
            \tag{A1} T_1[a_i]&=(1+(\nballot +3)(i-1))s_1 +(\nballot+2) s_{i+1} \Rightarrow (\nballot +3)i \geq T_1[a_i] >  (\nballot +3)(i-1) \label{eq:A1}\\
            \tag{A2} T_2[a_i]&=(1+(1/\varepsilon +1)(i-1))s_1 +(1/\varepsilon) s_{i+1} \Rightarrow   (1/\varepsilon +1)i \geq T_2[a_i] >  (1/\varepsilon +1)(i-1) \label{eq:A2} \\
            \tag{B1} T_1[b_i]&=(\nballot +3)(i-1)s_1 +(\nballot+2) s_{i} + s_\nballot \Rightarrow (\nballot +3)i > T_1[b_i] >  (\nballot +3)(i-1) \label{eq:B1}\\
            \tag{B2} T_2[b_i]&=(1/\varepsilon +1)(i-1)s_1 +(1/\varepsilon) s_{i}+s_\nballot \Rightarrow   (1/\varepsilon +1)i > T_2[b_i] >  (1/\varepsilon +1)(i-1)  \label{eq:B2}\\
            \tag{C1} T_1[c_j]&=(t(\nballot+3)+2j)s_1 \Rightarrow t(\nballot+3)+2j=T_1[c_j]>t(\nballot+3)+2(j-1)  \label{eq:C1}\\
            \tag{C2} \begin{split}
                T_2[c_j]&= ( 2|Z_j|+\nvars(1/\varepsilon+1) +6(\nballot+3)(j-1))s_1 + s_\nballot+2(\nballot+2)\left(\sum_{g \in V_j\setminus Z_j }s_{g} +\sum_{g \in Z_j }s_{g+1} \right)\\ 
                &\Rightarrow \nvars(1/\varepsilon+1) +6(\nballot+3) j \geq  T_2[c_j] > \nvars(1/\varepsilon+1) +6(\nballot+3)(j-1)
            \end{split}\label{eq:C2}\\
            \tag{D1} T_1[d_j]&=(t(\nballot+3)+2j-1)s_1 +s_k\Rightarrow t(\nballot+3)+2j>T_1[d_j]>t(\nballot+3)+2(j-1)  \label{eq:D1}\\
            \tag{D2} \begin{split}
                 T_2[d_j]&=  (1+ \nvars(1/\varepsilon+1) +6(\nballot+3)(j-1))s_1 + 2|Z_j|s_\nballot + 2(\nballot+2)\left(\sum_{g \in V_j \setminus Z_j }s_{g+1} +\sum_{g \in Z_j }s_{g} \right) \\ 
                &\Rightarrow \nvars(1/\varepsilon+1) +6(\nballot+3)j >  T_2[d_j] > \nvars(1/\varepsilon+1) +6(\nballot+3)(j-1)
            \end{split} \label{eq:D2}
        \end{align}
        where $V_j \subseteq [\nvars]$ and $Z_j \subseteq V_j$ indicate the indices of the variables that appear in and that appear negated in the clause $C_j$, respectively.
        We now show that $s$ gives perfect agreement between $\profile^{(1)}$ and $\profile^{(2)}$, \emph{i.e.}, $(T_1[g]-T_1[h])(T_2[g]-T_2[h])>0$ for all distinct pairs of $g,h \in  A \sqcup B \sqcup C \sqcup D \sqcup E$. We will proceed by a case by case analysis: 
        \begin{itemize}
            \item \textbf{Case 1:} $g \in A \sqcup B$, $h \in C \sqcup D$. By \cref{eq:A1,eq:B1}, we have $T_1[g] \leq t(\nballot+3)$, since $i \leq t$.  By \cref{eq:C1,eq:D1}, we have $T_1[h] > t(\nballot+3)$, as $j\geq 1$. This implies $T_1[g]<T_1[h]$. Similarly, by \cref{eq:A2,eq:B2} we have $T_2[g] \leq t(1/\varepsilon+1)$, and by \cref{eq:C2,eq:D2} we have $T_2[h] >  t(1/\varepsilon+1)$, as $j\geq 1$. This implies $T_2[g]<T_2[h]$. Hence, $(T_1[g]-T_1[h])(T_2[g]-T_2[h])>0$, as desired.
            
            \item  \textbf{Case 2:} $g \in \{a_i,b_i\}$, $h \in \{a_j,b_j\}$ for some $t\geq i>j\geq 1$.  By \cref{eq:A1,eq:B1}, we have $T_1[g] > (\nballot+3)(i-1) \geq (\nballot+3)j$ and  $T_1[h] \leq (\nballot+3)j$. This implies $T_1[g] > T_1[h]$. Similarly, by \cref{eq:A2,eq:B2} we have $T_2[g] > (1/\varepsilon+1)(i-1) \geq (1/\varepsilon+1)j$ and  $T_2[h] \leq  (1/\varepsilon+1)j$. This implies $T_2[g]>T_2[h]$. Hence, $(T_1[g]-T_1[h])(T_2[g]-T_2[h])>0$, as desired.
            
            \item \textbf{Case 3:} $g = a_i$, $h=b_i$ for some $i \in [t]$. In this case $T_1[g]-T_1[h]=1-(\nballot+2)(s_i - s_{i+1}) = 1-(\nballot+2)\delta_i$ and $T_2[g]-T_2[h]=1-(1/\varepsilon)(s_i - s_{i+1}) = 1-\delta_i/\varepsilon$. If $\delta_i=\frac{\varepsilon}{2}$ (\emph{i.e.} $x^*_i$ is False), $(T_1[g]-T_1[h])(T_2[g]-T_2[h])=(1-\frac{1}{14})(1-\frac{1}{2})>0$. If $\delta_i=\frac{1}{\nballot+1}$ (\emph{i.e.} $x^*_i$ is True), $(T_1[g]-T_1[h])(T_2[g]-T_2[h])=(1-\frac{\nballot+2}{\nballot+1})(1-\frac{7(\nballot+2)}{\nballot+1})= (\frac{-1}{\nballot+1})(\frac{-6\nballot-13}{\nballot+1})>0$.

            \item \textbf{Case 4:} $g \in \{c_i,d_i\}$, $h \in \{c_j,d_j\}$ for some $t\geq i>j\geq 1$. By \cref{eq:C1,eq:D1}, we have $T_1[g] > t(\nballot+3)+2(i-1) \geq  t(\nballot+3)+2j$ and  $T_1[h] \leq t(\nballot+3)+2j$. This implies $T_1[g] > T_1[h]$. Similarly, by \cref{eq:C2,eq:D2} we have $T_2[g] > t(1/\varepsilon+1)+6(\nballot+3)(i-1) \geq   t(1/\varepsilon+1)+6(\nballot+3) j$ and  $T_2[h] \leq t(1/\varepsilon+1)+6(\nballot+3) j$. This implies $T_2[g] > T_2[h]$. Hence, $(T_1[g]-T_1[h])(T_2[g]-T_2[h])>0$, as desired.
            
            \item \textbf{Case 5:} $g = c_i$, $h=d_i$ for some $i \in [\nclause]$. In this case $T_1[g]-T_1[h]=1$ and 
            \begin{align*}
                 T_2[g]-T_2[h]&=2|Z_i|-1+2(\nballot+2)\left(\sum_{g \in V_i \setminus Z_i}(s_{g}-s_{g+1})-\sum_{h \in Z_i  }(s_{h}-s_{h+1}) \right)\\
                 &=2|Z_i|-1+2(\nballot+2)\left(\sum_{g \in V_i \setminus Z_i}\delta_g-\sum_{h \in Z_i  }\delta_h \right),
            \end{align*}
            where $V_i \subseteq [\nvars]$ and $Z_i \subseteq V_i$ indicate the indices of the variables that appear in and that appear negated in the clause $C_i$, respectively. Since $\{x^*_i\}_{i \in \nvars}$ is a satisfying assignment by assumption, we must either have $\delta_{g'} =\frac{1}{\nballot+1}$ (\emph{i.e.} $x^*_{g'}$ is True) for some $g' \in V_i \setminus Z_i$ or $\delta_{h'}=\frac{\varepsilon}{2}$ (\emph{i.e.} $x^*_{h'}$ is False) for some $h' \in Z_i$. If the former is true ($\delta_{g'} =\frac{1}{\nballot+1}$ for some $g'\in V_i \setminus Z_i$):
            \begin{align*}
                 T_2[g]-T_2[h]&=2|Z_i|-1+2(\nballot+2)\left(\frac{1}{\nballot+1}+\sum_{g \in V_i \setminus (Z_i\cup\{g'\} )}\delta_g-\sum_{h \in Z_i  }\delta_h \right) 
                 \\& \geq 2|Z_i|-1+2(\nballot+2)\left(\frac{1}{\nballot+1}+\frac{(|V_i|-|Z_i|-1)\varepsilon}{2}-\frac{|Z_i|}{\nballot+1} \right) \\&\geq 2|Z_i|-1+2(\nballot+2)\left(\frac{1-|Z_i|}{\nballot+1}-\frac{\varepsilon}{2} \right)  \\&=-\frac{2|Z_i|}{\nballot+1}+\frac{\nballot+3}{\nballot+1} - \frac{1}{7}= \frac{-2|Z_i| +\frac{6}{7}\nballot + \frac{20}{7}}{\nballot+1}.
            \end{align*}
            Similarly, if the latter is true ($\delta_{h'} =\frac{\varepsilon}{2}$ for some $h'\in Z_i$):
            \begin{align*}
                 T_2[g]-T_2[h]&=2|Z_i|-1+2(\nballot+2)\left(\sum_{g \in V_i \setminus Z_i}\delta_g-\frac{\varepsilon}{2}-\sum_{h \in Z_i \setminus \{h'\}  }\delta_h \right) 
                 \\& \geq 2|Z_i|-1+2(\nballot+2)\left(-\frac{\varepsilon}{2}-\frac{|Z_i|-1}{\nballot+1} \right)=  -\frac{2|Z_i|}{\nballot+1}+\frac{\nballot+3}{\nballot+1}-\frac{1}{7}\\&= \frac{-2|Z_i| +\frac{6}{7}\nballot + \frac{20}{7}}{\nballot+1},
            \end{align*}
            which gives the same inequality. Consider two cases: if $\nvars \geq 2$ (\emph{i.e.}, $\nballot=\nvars +2 \geq 4$), we get
            \begin{align*}
                   T_2[g]-T_2[h] \geq   \frac{-2|Z_i| +\frac{6}{7}\nballot + \frac{20}{7}}{\nballot+1}  \geq \frac{-6 +\frac{24}{7} + \frac{20}{7}}{\nballot+1}= \frac{\frac{2}{7}}{\nballot+1} >0,
            \end{align*}
            since $|Z_i| \leq 3$. If $\nvars =1$ (\emph{i.e.}, $\nballot=3$), then $|Z_i|\leq 1$, since $\nvars$ is the number of variables in $\phi$. Then,
            \begin{align*}
                   T_2[g]-T_2[h] \geq   \frac{-2|Z_i| +\frac{6}{7}\nballot + \frac{20}{7}}{\nballot+1}  \geq \frac{-2 +\frac{18}{7} + \frac{20}{7}}{\nballot+1}= \frac{\frac{24}{7}}{\nballot+1} >0.
            \end{align*}
            In both cases, we have $T_2[g]-T_2[h]>0$ and hence $(T_1[g]-T_1[h])(T_2[g]-T_2[h])>0$, as desired.
   
            \item \textbf{Case 6:} $g \in E\setminus \{e_\nballot\}, h \notin E$. Say $g=e_i$ for some $1\leq i< \nballot$ and $n_\set = \sum_{f \in  A \sqcup B \sqcup C \sqcup D }\sum_{j=1}^\nballot \cmatrix_\set[f,j]$ for $\set \in \onetwo$. Due to the last set of voters that we added while constructing $\profile^{(1)}$ and $\profile^{(2)}$ (those that rank only elements of $E$), we have
            \begin{align*}
            \cmatrix_\set[e_i,i] \geq 14(\nballot+2) \max_{\gamma \in \onetwo} n_{\gamma} &\geq 14(\nballot+2) \sum_{j=1}^\nballot \cmatrix_\set [h, j] \\
            &\geq 14(\nballot+2) \sum_{j=1}^\nballot \cmatrix_\set [h, j] s_j = 14(\nballot+2) T_\set[h]
            \end{align*} 
            Moreover, since $s_i \geq s_{\nballot-1} \geq 1-\frac{\nballot-2}{\nballot+1} =\frac{3}{\nballot+1} > \frac{1}{\nballot+1}$, we have
                $T_\set[g]= \cmatrix_\set[g,i]s_i > \frac{\cmatrix_\set[e_i,i]}{\nballot+1} > T_\set[h]$
            for each $\set \in \onetwo$. This gives $(T_1[g]-T_1[h])(T_2[g]-T_2[h])>0$, as desired.
            
            \item \textbf{Case 7:} $g=e_\nballot$. By construction $\cmatrix_\set[e_\nballot,j] > 0 \Rightarrow j=\nballot$, and hence $T_\set[g]=T_\set[e_\nballot]= \cmatrix_\set[e_\nballot,\nballot] s_\nballot=0$ for each $\set \in \onetwo$. If $h \in  A \sqcup B \sqcup C \sqcup D $, we have $T_\set[h]>0$ for both $\set \in \onetwo$ by \cref{eq:A1,eq:A2,eq:B1,eq:B2,eq:C1,eq:C2,eq:D1,eq:D2}. If $h=e_i$ for some $i \in [\nballot-1]$, on the other hand, $\cmatrix_\set[h,i]>0$ by the last set of voters added to $\profile^{(1)}$ and $\profile^{(2)}$ (those that only rank the elements of $E$), so once again we have $T_\set[h] = \cmatrix_\set[h,i]s_i>0$. Hence, we have $(T_1[g]-T_1[h])(T_2[g]-T_2[h])>0$, as desired.
            
            \item \textbf{Case 8:} $g=e_i$, $h=e_j$ for some $i<j<k$. Fix some $\set \in \onetwo$. By construction, $\profile^{(\set)}$ contains two types of rankings: (a) those that contain one elements in $ A \sqcup B \sqcup C \sqcup D $ and all remaining elements are those in $E$, of which there $n_\set = \sum_{f \in  A \sqcup B \sqcup C \sqcup D } \sum_{j=1}^\nballot \cmatrix_\set[f,j]$ many and (b) those that rank only and all elements in $E$, of which there are say $y_\set$ many, with $y_\set \geq \max_{\gamma \in \onetwo} 14(\nballot+2)n_\gamma \geq 14(\nballot+2)n_\set$. Hence, for any $j' \in [\nballot-1]$, we have $n_\set + y_\set > \cmatrix[e_{j'},j'] \geq y_\set$, where the first inequality is strict since $\cmatrix_\set[f,j']>0$ for some $f \in  A \sqcup B \sqcup C \sqcup D $ for all $j' \in [\nballot-1]$ (see, for example, the definitions for $\cmatrix_\set[a_{j'},j'+1]$ and  $\cmatrix_\set[b_{j'},j']$ for each $j' \in [\nvars]$), so there is at least some voters in the group of $n_\set$ that do not rank $e_{j'}$.
            This implies
            \begin{align*}
                T_\set[g]-T_\set[h]&= T_\set[e_i]-T_\set[e_j] = \cmatrix_\set[e_i,i]s_i-\cmatrix_\set[e_j,j]s_j \\& > y_\set s_i - (n_\set+y_\set)s_j = y_\set(s_i-s_j)-n_\set s_j \\
                & \geq y_\set(s_{i}-s_{i+1})-n_\set=y_\set \delta_i -n_\set \\&\geq 14(k+2)n_\set \cdot \frac{\varepsilon}{2}-n_\set =n_\set -n_\set = 0
            \end{align*}
        Hence, we have $(T_1[g]-T_1[h])(T_2[g]-T_2[h])>0$, as desired.
    \end{itemize}
We have shown that $(T_1[g]-T_1[h])(T_2[g]-T_2[h])>0$ for all $g,h \in \cand$, proving that $s$ indeed gives full agreement between $\profile^{(1)}$ and $\profile^{(2)}$.

\end{proof}
\section{Experiment Details}
\label{sec:appendix_experiments}

This appendix provides additional details on experiments discussed in \Cref{sec:experiments}.

\subsection{Ground Truth and Disagreement}\label{app:gt_v_dis}

In our first experiment, we evaluate the relationship between the rankings produced by SWFs and the ground truth ranking of a preference distribution.

Both Mallows and Plackett-Luce (PL) preference models naturally correspond to an ``ideal'' ranking. These distributions also have known maximum likelihood estimators: Mallow's MLE is the Kemeny function while we estimate the MLE of PL preferences using the \texttt{choix} Python library\footnote{https://choix.lum.li/en/latest/index.html}.

In our experiment, for each distribution we generate 50 elections with 100 voters and 100 alternatives. We set $\phi = 0.4$ for Mallow's preferences and $\alpha_i = e^{0.5(m-i)}$ for Plackett-Luce. Each voter provides a ranking over 10 alternatives, sampled from the relevant distribution, such that each alternative is ranked by 10 voters. This type of partial ranking aligns with a paper reviewing framework where conference organizers may wish to ensure that each paper is reviewed by a certain number of reviewers while each reviewer receives a certain number of assignments.

For each generated election, we do the following for several SWFs: Assign each voter into one of two groups, chosen uniformly at random for each voter. Each of these groups is then used as a complete profile to generate a weak ranking according to the SWF. We measure the $KT$ distance between these two rankings. We report the mean over these 10 splits as the ``split distance.'' We also calculate the $KT$ distance from the ranking generated by the SWF applied to all voters to the ground truth ranking (the ``ground truth distance'').

We plot each pair of distances for each SWF. In \autoref{fig:ground_truth_vs_central_ranking} the MLE of each noise model tends to minimize each distance. In general, there is a very strong relationship between the two distances. As the split distance increases, so too does the ground truth distance.

In this experiment the SWFs we evaluated were Kemeny, Plackett-Luce MLE (as implemented by the \texttt{choix} library), Borda Min-Max, and Optimized Positional Scores (see \autoref{sec:appendix_voting_definitions}) in addition to positional scoring rules with the following vectors:

 \begin{itemize}
     \item Plurality $(1, 0, 0, \ldots, 0)$
     \item Plurality + Veto $(1, 0.5, 0.5, \ldots, 0.5, 0)$
     \item Veto $(1, 1, \ldots, 1, 0)$
     \item Two-Approval $(1, 1, 0, \ldots, 0)$
     \item Borda $(\frac{m-1}{m-1}, \frac{m-2}{m-1}, \ldots, \frac{m-m}{m-1})$
 \end{itemize}

\subsection{Axiom Violations}
\label{app:axiom_violation}

\begin{figure*}[t]
    \centering
    \includegraphics[width=\linewidth]{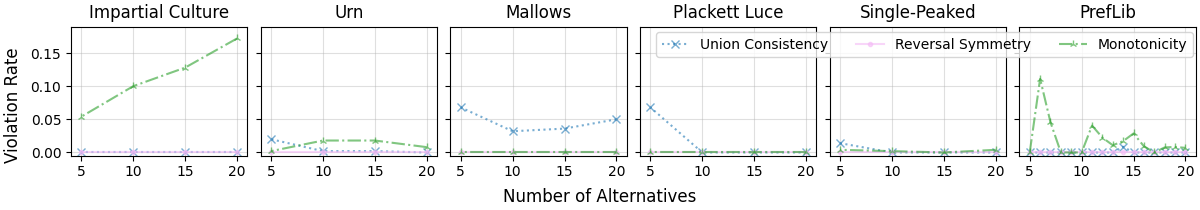}
    \caption{
    Violation rate of axioms across several preference distributions and real-world election data (rightmost). 
    Axiom violations generally decrease as the number of alternatives increases.
    }
    \label{fig:axiom_violation_rates_appendix}
\end{figure*}

To further explore the theoretical results of \Cref{sec:axiom} we analyze experimentally how often certain axioms are violated across several distinct preference distributions. Results of this experiment are found in \autoref{fig:axiom_violation_rates_appendix}. We measure the axiom violation rate, as defined by \citet{caiata2025}; an experimental measure of how often axioms are violated in practice on given preference distributions and voting rules.
For this experiment we sample, for each distribution, 500 profiles with 100 voters that each provide a full ranking over $m \in {5, 10, 15, 20}$ alternatives from the Impartial Culture, Urn (with $\alpha$ sampled from a Gamma distribution with shape $k = 0.8$ and scale $\theta = 1$ \cite{Boehmer21:Putting}), Mallows ($\phi = 0.4)$, Plackett-Luce ($\alpha_i = e^{0.5(m-i)}$), and Single-Peaked preference distributions \cite{Brandt16:Handbook}. Additionally, we consider all elections with complete preferences from PrefLib with up to 1000 voters and $m$ alternatives for $5 \leq m \leq 20$ (a total of 1392 elections) \cite{Mattei13:Preflib}. This real-world election data is compiled from a wide variety of sources with varying underlying preference distributions. 

On each election we take 50 random splits and calculate the fraction of splits on which each axiom is violated by a Rule Picking Rule using several positional scoring rules. 

As we evaluate the Reversal Symmetry axiom we are constrained to using positional scoring rules in our Rule Picking Rule. For all axioms, our RPR uses the following rules: Plurality, Plurality + Veto, Veto, Two-Approval, and Borda count. The scoring vector associated with each of these rules is given in \Cref{app:gt_v_dis}.

We provide an example to illustrate how we measure axiom violations.

\textit{Example.} Consider the monotonicity axiom. In a profile $\profile$ and a modification $\profile_a$ where some voters increase the rank they give to alternative $a$. Monotonicity is violated if the rank of $a$ under the RPR is lower in $\profile_a$ than in $\profile$.

We test for a violation of monotonicity in a profile $\profile$ by selecting the alternative $a$ ranked first in $\profile$ by the RPR, selecting some uniform random fraction of voters in (0.2, 0.8) to increase their ranking of $a$. If the RPR on $\profile_a$ assigns a different rank to $a$ we say that the axiom has been violated. We report the fraction of instances tested in which the axiom is violated.

As shown in \Cref{prop:reverse} (which easily extends to our implementation using Monte Carlo sampling), reversal symmetry is never violated and only checked as a test. Moreover, under all but one distribution, violations of union consistency and monotonicity are rare ($<0.05$ of the time for $\geq 10$ alternatives) and generally decrease as the number of alternatives increases.\footnote{The one exception to this trend is monotonicity in the Impartial Culture (IC) distribution, which is not centered around a ground truth. One possible explanation for this is IC maximizing the probability of majority cycles~\citep{Tsetlin03:Impartial}, under which monotonicity violations are more likely to occur.}

\subsection{Alternative Distance Functions}
\label{app:jaccard}

\begin{figure}[t]
    \centering
    \begin{subfigure}[t]{0.49\textwidth}
        \centering
        \includegraphics[width=\textwidth]{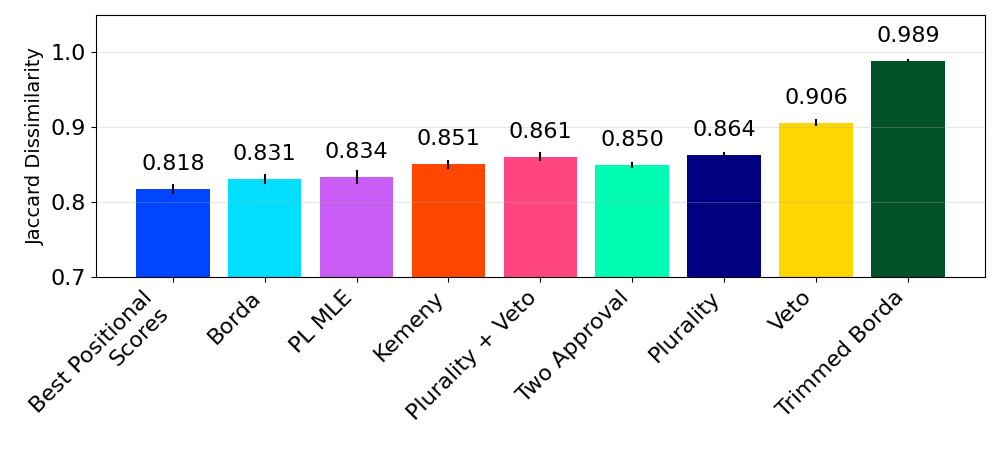}
        \caption{
        Jaccard dissimilarity between sets of winners for splits of partial rankings of ALMA Cycle 10 project proposals.
        } 
        \label{fig:alma_cycle10-jaccard}
    \end{subfigure}
    \hfill
    \begin{subfigure}[t]{0.49\textwidth}
        \centering
        \includegraphics[width=\textwidth]{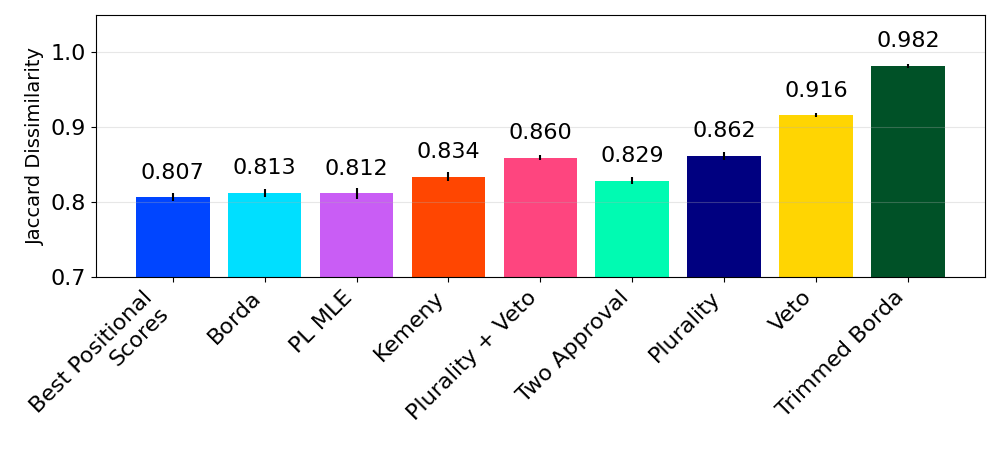}
        \caption{
        Jaccard dissimilarity between sets of winners for splits of partial rankings of ALMA Cycle 11 project proposals.
        }
        \label{fig:alma_output-jaccard}
    \end{subfigure}
    \caption{Jaccard dissimilarity and standard error of several rules on sets of winners using partial rankings over project proposals for the Atacama Large Millimeter Array (ALMA). Winners of a rule are the first 240 proposals in the ranking output by that rule.
    } 
    \label{fig:alma_jaccard_plots}
\end{figure}

As mentioned in \Cref{sec:aba}, the $\aba$ framework can also be used in conjunction with distance functions other than the (weighted) Kendall-Tau distance. This can be especially for settings such as peer review of poroposals for the Atacama Large Millimeter Array (ALMA) \citep{Meyer22:Analysis}. Here, the goal is not necessarily to rank all proposals, but to pick $k$ proposals to be funded, for some $k < \ncand$. It may thus be more useful to study the consistency the top $k$ proposals picked by each rule. Indeed, we explore this possibility by using the \textit{Jaccard dissimilarity} \citep{jaccard1901etude} to measure the similarity of winning proposals on ALMA data. 

Jaccard dissimilarity measures the overlap of two sets; it is defined as the ratio of the symmetric difference to the size of the union of two sets. For each voting rule that we evaluate, we select the 240 highest ranked proposals as the winning proposals (chosen to match the number of winning proposals chosen in ALMA Cycle 10).

In \Cref{fig:alma_jaccard_plots} we see the Jaccard dissimilarity for several rules on proposal rankings from Cycle 10 and 11 of ALMA. The differences between rules have a similar relative order to the Kendall-Tau distance of the rules applied to the same data (\Cref{fig:alma_bar_plots}).

\subsection{City Election Data}\label{sec:city-election}

In \autoref{fig:empirical_scatter} we show results of several rules on Preflib data from several real-world political elections. Each of these elections used Instant Runoff Voting (IRV) to compute the empirical winner. Given a profile, IRV iteratively eliminates the alternative with the least plurality score (\emph{i.e.}, with the least number of voters ranking them top) by removing them from the profile, until a single alternative remains. As an SWF, we interpret IRV as the rule outputing the alternatives in the reverse of this elimination order (the alternative eliminated first is ranked bottom, and so on). We provide details on each election shown in \autoref{tab:preflib_city_election_details}. The data from Preflib can be freely distributed and used under the GPL-3.0 license.

Initially we collected all elections included in Preflib Elections 5, 16, 17 ,18, 19, 20, 21, and 22. However, in \autoref{fig:empirical_scatter} we have filtered out two types of election from our starting data data in order to make the plot more readable:

\begin{itemize}
    \item Any elections where all rules had a split distance of $0$. While this is a random event, we found that this consistently included a set of 10 elections.
    \item Any election with more than $100$ candidates. This excluded two specific elections that occurred in 2009 in the city of Minneapolis with 379 and 477 alternatives. These two elections have exceptionally high split distances and were removed as outliers due to their split distance and unusual number of candidates.
\end{itemize}

\begin{table}[ht]
\centering
\begin{tabular}{@{}cp{1.6cm}p{5cm}cc@{}}
\toprule
Plot Index & Preflib Election ID & Election Name                                                                            & \# Voters & \# Candidates \\ \midrule
0          & 5                   & City (2009 Burlington Mayoral Election)                                                  & 8980   & 6  \\
1          & 16                  & City (Aspen City Council 2009)                                                           & 2477   & 11 \\
2          & 16                  & City (Aspen Mayor 2009)                                                                  & 2527   & 5  \\
3          & 17                  & City (2010 Berkeley City Council - District 7)                                           & 4173   & 4  \\
4          & 18                  & City (2009 Minneapolis Board of Estimate and Taxation Election - No Write In)            & 32086  & 7  \\
5          & 18                  & City (2009 Minneapolis Park and Recreation Commissioner At-Large Election - No Write In) & 36655  & 9  \\
6          & 19                  & City (2010 Oakland Mayor)                                                                & 119256 & 11 \\
7          & 19                  & City (2010 Oakland City Council - District 4)                                            & 20981  & 8  \\
8          & 19                  & City (2012 Oakland City Council - District 3)                                            & 22079  & 7  \\
9          & 19                  & City (M2012 Oakland City Council - District 1)                                           & 28660  & 8  \\
10         & 20                  & City (2008 Pierce County Assessor - Treasurer)                                           & 262312 & 7  \\
11         & 21                  & City (2011 San Francisco Mayor)                                                          & 194530 & 25 \\
12         & 21                  & City (2011 San Francisco Sheriff)                                                        & 183192 & 5  \\
13         & 21                  & City (2012 San Francisco Board of Supervisors - District 5)                              & 35183  & 9  \\
14         & 21                  & City (San Francisco Board of Supervisors - District 7)                                   & 31437  & 10 \\
15         & 21                  & City (2010 San Francisco Board of Supervisors - District 2)                              & 24109  & 7  \\
16         & 21                  & City (2008 San Francisco Board of Supervisors - District 9)                              & 26634  & 8  \\
17         & 21                  & City (2008 San Francisco - Board of Supervisors District 3)                              & 27310  & 10 \\
18         & 21                  & City (2008 San Francisco Board of Supervisors - District 1)                              & 28777  & 10 \\
19         & 21                  & City (2010 San Francisco Board of Supervisors - District 10)                             & 18001  & 22 \\
20         & 21                  & City (2008 San Francisco Board of Supervisors - District 11)                             & 24717  & 10 \\
21         & 21                  & City (2010 San Francisco Board of Supervisors - District 6)                              & 21188  & 15 \\
22         & 22                  & City (2010 San Leandro Mayor)                                                            & 22407  & 7  \\
23         & 22                  & City (2012 San Leandro City Council - District 4)                                        & 23236  & 5  \\
24         & 22                  & City (2012 San Leandro City Council - District 2)                                        & 25355  & 4  \\ \bottomrule
\end{tabular}
\caption{Details on the election represented at each index in \autoref{fig:empirical_scatter}.}
\label{tab:preflib_city_election_details}
\end{table}

\subsection{Formula One Data}\label{app:formulaone}

\begin{figure}[t]
    \centering
    \includegraphics[width=\textwidth]{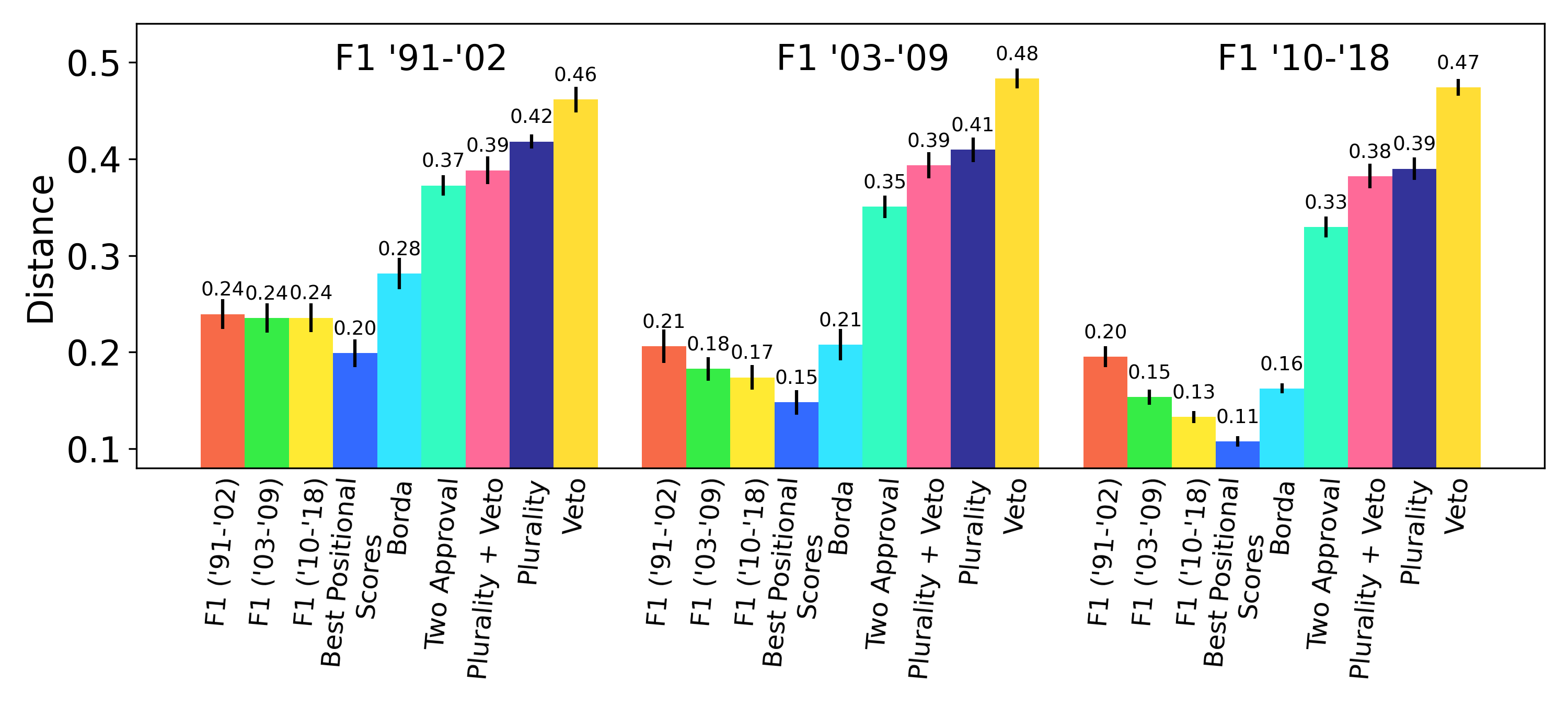}
    \caption{Distance between splits for SWFs aggregating rankings of drivers in F1 races. Each rule is evaluated on each period of races. Newer F1 rules provide lower distance on all race periods.
    }
    \label{fig:f1_all_rules-appendix}
\end{figure}

The F1 portion of \autoref{fig:olympic_f1_splits} is generated using data about races found on Preflib \cite{Mattei13:Preflib}; specifically, we use the complete form of dataset ID 00053.
This contains one profile for each Formula One season with a preference order corresponding to each individual race in the season. Each preference order lists the drivers that competed in \textit{every} race in that season, ordered by the position in which they finished that race. 

We show in \autoref{fig:f1_all_rules-appendix} the $KT$ distance for all rules divided by racing period. While the rules with highest distance (Two-Approval, Plurality + Veto, Plurality, Veto) stay quite consistent, all other rules provider lower $KT$ distances on more recent race periods. Notably, the F1 rules themselves become much more consistent over time.

\subsection{Olympic Data}\label{app:olympics}

\begin{figure}[t]
    \centering
    \includegraphics[width=.7\linewidth]{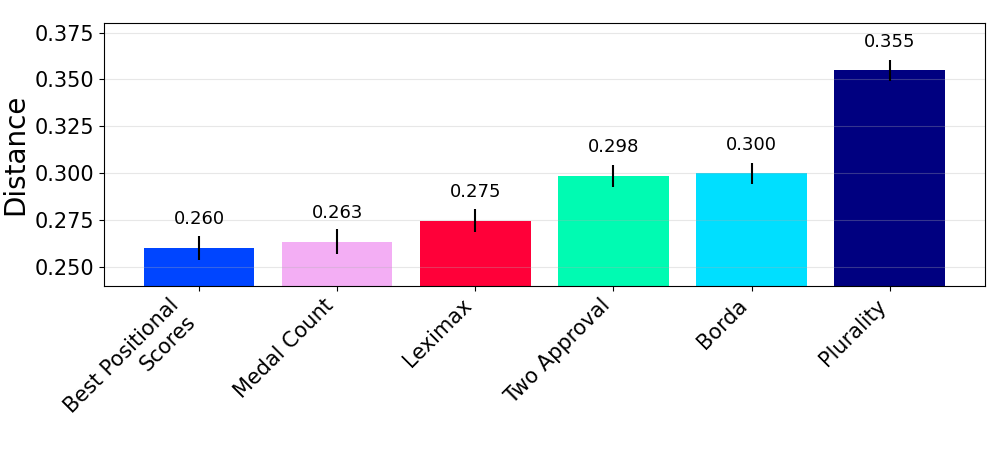}
    \caption{Average distance and standard error between splits for SWFs over rankings induced by Olympic medals. Optimization of positional scores is only occasionally able to improve upon the ranking generated by giving one point to each country for each medal they win, regardless of medal type.}
    \label{fig:olympics_bar_appendix}
\end{figure}

To evaluate rules on Olympic medal wins we use a Kaggle dataset providing details of all Olympic results (Summer and Winter) between 1896 and 2016 \citep{kaggle:olympics}. From this we extract the winning countries of each medal for each event. We convert each event into a partial ranking by assigning to first place all countries receiving gold medals, to second place all countries receiving silver medals, and to third place all countries receiving bronze medals. Countries that did not win medals or did not compete in an event are not included in a ranking.

In the large majority of cases this results in a partial order consisting of exactly three countries, each in a different rank. In rare cases a position might be empty or have multiple winners (e.g. in 1992 Canada and USA won Gold in women's solo synchronized swimming, no Silver was awarded, and Bronze was won by Japan), or a country might occur multiple times (e.g. in 2008 Jamaica won Gold and two Silver medals in the women's 100 metre; no Bronze was awarded).
In these exceptional cases we do nothing different and award that country points for each of the positions that it occupies.

This results in one ``election'' for each year in which the Olympics occurred where preference orders correspond to partial rankings induced by medal wins. To these rankings we apply the $\aba$ framework as we have in all other experiments: We generate splits by randomly placing the profile induced by each event into one split or the other, then find the distance between splits and report the average over many sets of splits. Here we add two new rules:

\begin{itemize}
    \item \textbf{Medal Count:} Each medal is treated equivalently. Each country receives a point each time it appears in a preference order, regardless of rank.
    \item \textbf{Leximax:} Rank countries by the number of Gold medals they receive, breaking ties by counting Silver medals, breaking remaining ties by counting Bronze medals. In practice we implement this as the score vector $(1000000, 1000, 1)$.
\end{itemize}

Results of this analysis are displayed in \Cref{fig:olympics_bar_appendix}. We see that simply counting the total number of medals each country receives results in much lower disagreement than standard voting rules, only very rarely do optimized scores improve upon the Medal Count rule.

\subsection{Additional Definitions of Voting Rules}
\label{sec:appendix_voting_definitions}

Through our experiments we use several rules which are not fully described. In this section we describe each rule found throughout our paper.

\subsubsection{Kemeny}\label{app:kemeny}

Kemeny's rule is defined as the ranking which minimizes the sum of Kendall-Tau distances between the output ranking and each voter's individual ranking \cite{Kemeny59:Mathematics}. We implement the Kemeny method code provided by \citet{Baharev21:Exact} and the Gurobi optimization library \citep{Baharev21:Exact,Gurobi24}.

\subsubsection{Trimmed Borda}

Calculate the number of voters ranking each alternative at each rank. For each alternative $a_i$, remove one of the highest rankings which $a_i$ has received and remove one of the lowest rankings which $a_i$ has received. After all removals are complete, calculate Borda scores as normal \citep{Meyer22:Analysis}.

\subsubsection{Best Positional Scores}\label{app:bestpos}

This method uses the \texttt{optimal-voting} package generate a positional scoring vector which minimizes $KT$ distance between splits on a given profile \citep{armstrong2025optimal}. This package uses simulated annealing to generate novel positional scoring rules which optimize a target function.
In each of our experiments, we run optimization multiple times; starting once from the initial score vector of each other positional scoring rule used in the experiment. For instance, if we were comparing with Borda's rule and Plurality we run optimization twice, starting from $(m-1, m-2, ..., 0)$ (Borda), and $(1, 0, ..., 0)$ (Plurality). In all cases, we run annealing for 500 steps.
At each step, \texttt{optimal-voting} updates one index in the state vector (\textit{i.e.}, the positional score vector) by some amount sampled uniformly at random from $(0.05, 1)$. We restrict updates to those that result in the state vector being weakly monotonically decreasing (the value at each index is not higher than the previous value).
Using the updated state vector as a positional scoring vector we calculate the mean $KT$ split distance over each split of voters. We use the same set of splits in annealing as we do when evaluating each other voting rule.
The new state is accepted with a probability related to the magnitude of difference between the current and previous mean split distances, and the number of steps that have already occurred. A small magnitude of difference is more likely to be accepted earlier than later. States with lower mean split distance are always accepted. 

\subsubsection{Leximax}

This rule is used exclusively in evaluating Olympic data. For a given profile, alternatives are ranked according to the number of gold medals they have received. To break ties, tied alternatives are ordered according to the number of silver medals they have received. To break subsequent ties, tied alternatives are ordered according to the number of bronze medals they have received. Note that we can instantiate Leximax as a positional scoring rule where each position is \textit{much} larger than the subsequent position. In our Olympic experiments we use the (pre-normalization) vector $(1000000, 1000, 1)$. As there are never more than 1000 opportunities for a single country to receive a medal of one type this is equivalent to Leximax.

\subsubsection{Medal Count}

This rule is used exclusively in evaluating Olympic data. For a given profile, alternatives are ranked purely based on the number of medals they received with no regard for the type of medals.

Our Olympics data uses as input the partial rankings containing only countries winning medals in an event. In this case, the Medal Count rule is equivalent to the positional scoring rule with the vector $(1, 1, 1)$.

\end{document}